\documentclass[final,twoside,11pt]{entics}
\pdfoutput=1
\usepackage{amsmath, amssymb, amscd, tikz, tikz-cd, mathtools, enumitem}
\usepackage{mathrsfs, bm, stmaryrd, comment, url}
\usepackage{graphicx, hyperref}
\usepackage{enticsmacro}

\mathtoolsset{showonlyrefs=true}
\usetikzlibrary{shapes.geometric, trees}

\newcommand{\sem}[1]{\llbracket #1 \rrbracket}

\newcommand{\lrangle}[1]{\langle #1 \rangle}

\newcommand{\vsp}{}
\renewcommand{\-}{\mathchar`-}
\newcommand{\cat}{\mathit{Sub}}
\newcommand{\Nat}{\mathbb{N}}
\newcommand{\mathword}[1]{\mathit{#1}}
\newcommand{\mathope}[1]{\mathit{#1}}

\makeatletter
\renewcommand{\thefigure}{\thesection.\arabic{figure}}
\@addtoreset{figure}{section}
\makeatother

\def\conf{MFPS 2022} 	
\volume{1}			
\def\volu{1}			
\def\papno{17}			
\def\mydoi{{\normalshape  \href{https://doi.org/10.46298/entics.10374}{10.46298/entics.10374}}}
\def\copynames{T. Yoshida} 
\def\runtitle{Continuous Functions on Final Comodels of Free Algebraic Theories}

\def\lastname{Yoshida}

\begin{document}

\begin{frontmatter}
\title{Continuous Functions on Final Comodels\\ of Free Algebraic Theories}
\author{Tomoya Yoshida\thanksref{myemail}}
\address{Research Institute for Mathematical Sciences\\ Kyoto University\\ Kyoto, Japan}
\thanks[myemail]{Email: \href{mailto:ytomoya@kurims.kyoto-u.ac.jp} {\texttt{\normalshape
ytomoya@kurims.kyoto-u.ac.jp}}}

\begin{abstract}
    In 2009, Ghani, Hancock and Pattinson
    gave a tree-like representation of
    stream processors $A^{\Nat} \rightarrow B^{\Nat}$.
    In 2021, Garner showed that
    this representation can be established in terms of {\it algebraic theory} and {\it comodels}:
    the set of infinite streams $A^{\Nat}$ is the {\it final comodel} of
    the {\it algebraic theory of $A$-valued input} $\mathbb{T}_A$
    and the set of stream processors $\mathword{Top}(A^{\Nat},B^{\Nat})$
    can be seen as the {\it final $\mathbb{T}_A$-$\mathbb{T}_B$-\it bimodel}.
    In this paper, we generalize Garner's results to the case of {\it free algebraic theories}.
\end{abstract}
\end{frontmatter}

\section{Introduction}
\label{sec:1}

Writing and verifying programs handling infinite objects
such as {\it streams} and infinite trees are highly non-trivial tasks.
To ease it, many attempts to identify the mathematical principles behind
infinite computation on infinite data structures have been made.
Among them, a most active and well-developed areas are the theory of coalgebras,
in which infinite objects are captured as elements of final coalgebras
which enjoy nice universality and useful principles for programming and verification.
This work can be seen as a contribution in this direction,
making use of the recent development of algebraic theories and their comodels.

\subsection{Background and Our Result}

In computer science, coalgebras for functors appear in various ways
\cite{jacobs_2016,Jacobs1997ATO,ruttenUniversalCoalgebraTheory2000};
one is as the way of implementations of infinite data structures,
which appear as elements of final coalgebras.
The simplest example of infinite data structures is
infinite sequences (often called streams).
For a set $A$,
the set $A^{\Nat}$ of $A$-valued streams is a coalgebra for an endofunctor $A \times (\_)$ on $Set$
with a coalgebra structure map $\alpha:A^{\Nat} \rightarrow A \times A^{\Nat}$ given by
\begin{equation}
    a_0a_1a_2 \cdots \mapsto (a_0, a_1a_2 \cdots).
\end{equation}
Moreover, as is well known, the coalgebra $(A^{\Nat},\alpha)$ is
indeed the final coalgebra for the functor $A \times (\_)$,
i.e., for every $(A \times (\_))$-coalgebra $(X,\phi:X \rightarrow A \times X)$,
there exists a unique $(A \times (\_))$-coalgebra homomorphism $(X,\phi) \rightarrow (A^{\Nat},\alpha)$.
For an endofunctor $F$, we write $\nu X.F(X)$ for the final $F$-coalgebra when it exists.
Then, $A^{\Nat}$ can be denoted as $\nu X.A \times X$.

Here, we are interested in functions that translate infinite data to infinite data,
e.g., functions from $A$-streams to $B$-streams ($A$ and $B$ are sets).
One might say stream processors to mean all functions $A^{\Nat} \rightarrow B^{\Nat}$,
but we prefer to define stream processors as
{\it productive} functions $A^{\Nat} \rightarrow B^{\Nat}$.
Productivity of a map $f:A^{\Nat} \rightarrow B^{\Nat}$ in this context means that
$f$ will read only finite information of a input stream $\overrightarrow{a}$
to decide finite information of the output stream $f(\overrightarrow{a})$.
This constraint is reasonable from a computational perspective
because programs can read only finite information of infinite data.
Thus if a program uses the stream $f(\overrightarrow{a})$
then it in fact uses a finite segment of this stream
and the finite segment of $f(\overrightarrow{a})$
should be computed using finite information of $\overrightarrow{a}$.

As described in \cite{DBLP:journals/corr/abs-0905-4813,garnerStreamProcessorsComodels2021},
this productivity can be characterized as the continuity of functions.
The set $A^{\Nat}$ has a natural topology,
the product topology of discrete spaces $A$,
and the set of productive maps $A^{\Nat} \rightarrow B^{\Nat}$ coincides with
the set of continuous maps $A^{\Nat} \rightarrow B^{\Nat}$.
Though this characterization may be mathematically elegant,
by computational requirements, we want a coalgebraic characterization of these continuous functions.

There already have been several studies on coalgebraic representation of
continuous functions between final coalgebras
\cite{DBLP:journals/corr/abs-0905-4813,article,10.1145/3064899.3064900}.
In particular, Ghani, Hancock and Pattinson \cite{DBLP:journals/corr/abs-0905-4813}
studied stream processors $A^{\Nat} \rightarrow B^{\Nat}$ and
showed that they can be represented as elements (trees) of another final coalgebra
$\nu X.T_A(B \times X)$,
where $T_A(V)$ is the set of finite-depth $A$-ary branching trees with $V$-labelled leaves.
Let $I_{AB} \coloneqq \nu X.T_A(B \times X)$.
Then each element of $I_{AB}$ can be seen (coinductively) as a tree in $T_A(B \times I_{AB})$.
For example, when $A=\{0,1\}$, the following tree $t$ belongs to $I_{AB}$
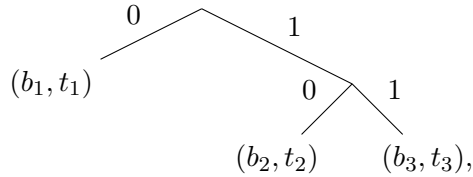
\[
    \begin{tikzpicture}[auto]
        \node (1) at (-2,-1) {$(b_1,t_1)$};
        \node (2) at (1,-2) {$(b_2,t_2)$};
        \node (3) at (3,-2) {$(b_3,t_3),$};
        \draw (0,0) --node[swap] {0} (1);
        \draw (0,0) --node {1} (2,-1) --node {1} (3);
        \draw (2,-1) --node[swap] {0} (2);
    \end{tikzpicture}
\]
where $b_i \in B$ and $t_i$'s are trees in $I_{AB}$ ($i=1,2,3$).
Here we give an overview of how $t$ expresses a function $\{0,1\}^{\Nat} \rightarrow B^{\Nat}$.
For a given stream $a_0a_1 \cdots \in \{0,1\}^{\Nat}$,
this tree $t$ outputs a $B$-valued stream stepwise as follows.
First, $t$ consumes $a_0$.
If $a_0=0$, it reaches the leaf $(b_1,t_1)$.
Thus it outputs $b_1$ and then the computation continues in the same way
with $t_1$ and $a_1a_2 \cdots$ (instead of $t$ and $a_0a_1 \cdots$).
On the other hand, if $a_0=1$, it does not reach leaves.
Therefore $t$ also consumes $a_1$ and reaches the leaf $(b_2,t_2)$ or $(b_3,t_3)$
if $a_1$ is $0$ or $1$, respectively.
Then the first output is $b_2$ or $b_3$ and computation continues
with either $t_2$ and $a_2a_3 \cdots$ or $t_3$ and $a_2a_3 \cdots$.
Eventually, $t$ will give a continuous map $\{0,1\}^{\Nat} \rightarrow B^{\Nat}$
because the above procedure produces each digit of the resulting stream
using a finite initial segment of the input stream.
However, this representation is {\it not bijective} in the sense that
different trees may give the same function.
For instance, consider trees $t$ and $t'$ such that
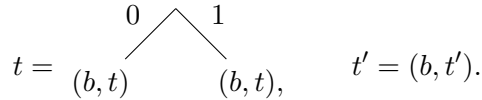
\[
    t=
    \begin{tikzpicture}[auto, baseline=-25pt]
        \node (l) at (-1,-1) {$(b,t)$};
        \node (r) at (1,-1) {$(b,t),$};
        \draw (0,0) --node[swap] {0} (l);
        \draw (0,0) --node {1} (r);
    \end{tikzpicture}
    \qquad
    t'=(b,t').
\]
Then both of them give a constant map which sends every stream to $bbb \cdots$
(in particular, the tree $t'$, which has no branches, does not consume elements to output the first digit $b$
and thus it does not consume an input stream at all to output the result $bbb \cdots$).

Garner \cite{garnerStreamProcessorsComodels2021} paid attention to this non-bijectivity
and he gave a bijective characterization of stream processors
in terms of {\it algebraic theories} and their {\it comodels}.
The set $A^{\Nat}$ arises as the {\it final comodel} of
the algebraic theory of {\it $A$-valued input} $\mathbb{T}_A$
and the set of stream processors $A^{\Nat} \rightarrow B^{\Nat}$ appears as
the {\it final $\mathbb{T}_A$-$\mathbb{T}_B$-bimodel}.
He also gave a comodel-theoretic expression of the result of \cite{DBLP:journals/corr/abs-0905-4813};
the final coalgebra $\nu X.T_{A}(B \times X)$ is the {\it final $\mathbb{T}_A$-residual $\mathbb{T}_B$-comodel}
and, because of their finalities, there are canonical maps between
the final $\mathbb{T}_A$-residual $\mathbb{T}_B$-comodel and the final $\mathbb{T}_A$-$\mathbb{T}_B$-bimodel
such that their composition is the identity function on the final $\mathbb{T}_A$-$\mathbb{T}_B$-bimodel.

In this paper, we generalize Garner's result to the case of {\it free} algebraic theories,
which are algebraic theories with no {\it equational axioms}.
We analyse continuous functions on final comodels of free algebraic theories,
investigate residual comodels and bimodels of free algebraic theories,
and relate them to each other.
If a free algebraic theory $\mathbb{T}$ has the {\it signature}
$\Sigma=\{\sigma_1,\ldots,\sigma_n\}$,
its final comodel $S_{\mathbb{T}}$ comprises infinite-depth $n$-ary branching trees
with labels determined by the signature $\Sigma$.
This is a generalization of the case of $A$-valued streams in the sense that
$A$-valued streams can be seen as infinite-depth $1$-ary branching trees with labels in $A$.
Thus continuous functions $S_{\mathbb{T}} \rightarrow S_{\mathbb{T'}}$ between
final comodels of free algebraic theories $\mathbb{T},\mathbb{T'}$ translate
trees to trees and they also observe
only finite information of input trees to decide finite information of output trees
(topologies on $S_{\mathbb{T}}$ and $S_{\mathbb{T'}}$ are defined in Section \ref{sec:3}).

The key point in this generalization is that these continuous functions can observe
{\it parallel} information such as sibling nodes in trees.
On the other hand, elements of the final $\mathbb{T}$-residual $\mathbb{T'}$-comodel can observe
only {\it serial} (in other words, {\it straight} or {\it successive}) finite information of input trees.
Therefore we must restrict ourselves to consider only such functions (we call them {\it straight functions} tentatively).
Here, as the case of streams, different elements can give the same straight function.
This non-bijectiveness will be eliminated when we consider the final $\mathbb{T}$-$\mathbb{T'}$-bimodel.
Consequently, we get the following diagram:
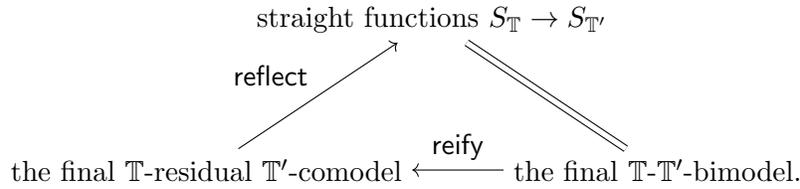
\begin{equation}
    \begin{tikzpicture}[auto]
        \node (s) at (0,0) {straight functions $S_{\mathbb{T}} \rightarrow S_{\mathbb{T'}}$};
        \node (i) at (-3,-2) {the final $\mathbb{T}$-residual $\mathbb{T'}$-comodel};
        \node (e) at (3,-2) {the final $\mathbb{T}$-$\mathbb{T'}$-bimodel.};
        \draw[double distance=2pt] (s) --(e);
        \draw[->] (e) --node[swap] {\textsf{reify}} (i);
        \draw[->] (i) --node {\textsf{reflect}} (s);
    \end{tikzpicture}
\end{equation}
And we will show that the composition $\textsf{reflect} \circ \textsf{reify}$
is the identity on the final $\mathbb{T}$-$\mathbb{T'}$-bimodel
when it is identified with the set of straight functions.
Although the restriction that we consider only straight functions might seem artificial,
the correspondence we will construct is in fact a generalization of \cite{garnerStreamProcessorsComodels2021}
because there is no parallel information in the case of streams.
Note that, to realize our restriction for maps,
we will consider a category $\cat$
instead of the usual category $\mathword{Top}$ of topological spaces and continuous functions.
Then the set of straight functions $S_{\mathbb{T}} \rightarrow S_{\mathbb{T'}}$
can be identified with a hom-set $\cat(S_{\mathbb{T}},S_{\mathbb{T'}})$.

\subsection{Outline of The Paper}

This paper is organized as follows.

In Section \ref{sec:2}, we review basic notions of algebraic theories, their models, and comodels.
Especially, we give examples of comodels to help with our intuition.

In Section \ref{sec:3}, we discuss topological notions on comodels.
Following Garner \cite{garnerStreamProcessorsComodels2021}'s {\it operational topology},
we define the {\it operational sub-basis}.
Since we mainly use sub-basis, we define the category $\cat$ of sets with sub-basis.
Additionally, we observe that final comodels in $\mathword{Set}$ become final comodels in $\cat$.

In Section \ref{sec:4}, we disucuss residual comodels and bimodels.
Each elements of the final residual comodel constructs a function between two final comodels.
This construction will be named {\textsf{reflect}}.
We also define a map {\textsf{reify}} from the final bimodel to the final residual comodel.

In Section \ref{sec:6},
we show the main result of this paper:
$\cat(S_{\mathbb{T}},S_{\mathbb{T'}})$ can appear as
the final $\mathbb{T}$-$\mathbb{T'}$-bimodel.
Though the proof is long and might look complex,
the technique and key ideas are similar to the case of stream processors \cite{garnerStreamProcessorsComodels2021}.
Difficulty comes from the existence of parallel information in trees,
which constitute final comodels of free theories.
At the end of this section, we explain that the composition $\textsf{reflect} \circ \textsf{reify}$ becomes the identity on $\cat(S_{\mathbb{T}},S_{\mathbb{T'}})$.

In Section \ref{sec:8},
we summarize our results and describe directions for future work.

\vspace{10pt}

We use the following notations.

\begin{itemize}
    \item $\mathbb{N} \coloneqq \{0,1,2,\ldots\}$ is the set of natural numbers.

    \item For a set $A$,
    $A^{\Nat}$ denotes the set of infinite sequences of elements of $A$ and
    $A^{\ast}$ denotes the set of finite sequences of elements of $A$.
    Their (disjoint) union $A^{\ast} \cup A^{\Nat}$ is denoted by $A^{\leq \Nat}$.
    The empty sequence in $A^{\ast}$ is written as $\epsilon$.

    \item For a finite sequence $s \in A^{\ast}$ and a finite or infinite sequence $t \in A^{\leq \Nat}$,
    $s \cdot t$ denotes the concatenation of $s,t$.

    \item Usually we use $V$, $W$ as sets of variables, and
    variables are denoted by $v$, $w$, $v'$, $v_1$, $v_2$, etc.
\end{itemize}

\subsection{Related Work}

References for (co)algebras include \cite{jacobs_2016,Jacobs1997ATO,ruttenUniversalCoalgebraTheory2000},
in which various usages of coalgebras in computer science are explained;
one is as implementations of infinite objects and
another one is as representations of transition systems.
On the other hand, comodels also can be seen as transition systems.
As described in \cite{powerComodelsCoalgebrasState2004,10.1016/j.entcs.2008.10.018,uustaluStatefulRunnersEffectful2015},
comodels of an algebraic theory $\mathbb{T}$ provide environments for evaluating $\mathbb{T}$-terms.
Applications of algebraic theories in the study of computational effects originated in
Plotkin and Power \cite{inproceedings}, and originally,
such categorical studies of computational effects were founded by Moggi \cite{moggiNotionsComputationMonads1991},
in which he advocated that various kinds of computational effects, such as exception and nondeterminism,
can be modeled by monads.

Researches on coalgebraic characterizations of continuous functions on final coalgebras
originated in Ghani et al. \cite{DBLP:journals/corr/abs-0905-4813},
which is about stream processors.
After that, in \cite{article}, they generalized their results to
the case of final coalgebras of functors called {\it containers}.
In their work, remaining problems were non-bijectiveness of representations
and verification of {\it completeness} in the latter case,
i.e., it is not resolved whether all continuous functions can be expressed by their representations
(however, for the case of stream processors, completeness was already proved).
Then Garner \cite{garnerStreamProcessorsComodels2021} gave a bijective characterization of stream processors
and reformulated the result of \cite{DBLP:journals/corr/abs-0905-4813}
in terms of algebraic theories and their comodels
(in particular, he used notions of residual comodels and bimodels).
Our study generalizes Garner's techniques to the case of free algebraic theories.
An advantage of comodel-theoretic characterization is
the easy verification of completeness, which will be simply done by the universality of final objects.
Therefore we expect that this technique is useful to give complete characterizations in more general cases.

As for stream processors,
there exists another well-known related concept called transducers \cite{10.1109/FOCS.1967.3},
which is a generalization of automata.
In \cite{bealDeterminizationTransducersFinite2002}, Beal and Carton
argued when functions $A^{\Nat} \rightarrow B^{\Nat}$
realized by transducers become continuous
(in their paper, $A$ and $B$ were assumed to be finite).
In a recent work \cite{lmcs:7712},
Hyvernat used type-theoretic transducers to represent continuous functions between coinductive types
(note that completeness of representations is not verified yet).
According to Hyvernat,
observing parallel information in trees is equivalent to the backtracking of transducers.
This insight leads us to an idea that,
if there is a comodel-theoretic characterization of backtracking,
we might be able to characterize not only straight functions
but also arbitrary continuous functions between final comodels.

\section{Algebraic Theories and Their (Co)Models}
\label{sec:2}

In this section, we review basic concepts and notations for
algebraic theories and their models as well as comodels.
We mostly follow Garner's treatment in \cite{garnerStreamProcessorsComodels2021,garnerCostructurecosemanticsAdjunctionComodels2020}.

\begin{definition}
    An {\it algebraic theory} $\mathbb{T}$ is a pair $(\Sigma_{\mathbb{T}},E_{\mathbb{T}})$,
    where $\Sigma_{\mathbb{T}}$ is a {\it signature} and $E_{\mathbb{T}}$ is a set of {\it equations} over $\Sigma_{\mathbb{T}}$.
    A {\it signature} comprises a set $\Sigma$ of operation symbols, and for each $\sigma \in \Sigma$ a set $|\sigma|$, its {\it arity}.
    Given a signature $\Sigma$ and a set $V$, we define the set $\Sigma(V)$ of $\Sigma${\it -terms with variables in} $V$ by
    the inductive clauses
    \begin{align}
        v \in V &\Rightarrow v \in \Sigma(V),\\
        \sigma \in \Sigma,\  t_i \in \Sigma(V)\ (i \in |\sigma|) &\Rightarrow
        \sigma(\lambda i.t_i) \in \Sigma(V).
    \end{align}
    An {\it equation over a signature} $\Sigma$ is a formal equality $t=u$ between terms in the same set of free variables.
    We say $\mathbb{T}$ is {\it free} if it has no equation, i.e., if $E_{\mathbb{T}}=\emptyset$.
\end{definition}

We usually say ``theory'' to mean ``algebraic theory''.

\begin{definition}
    For a signature $\Sigma$ and a term $t \in \Sigma(V)$ and terms $u_v \in \Sigma(W)$ $(v \in V)$,
    we define the {\it substitution} $t(\lambda v.u_v) \in \Sigma(W)$ by recursion on $t$:
    \begin{align}
        v \in V &\Rightarrow v(\lambda v.u_v) \coloneqq u_v,\\
        \sigma \in \Sigma,\ t_i \in \Sigma(V)\ (i \in |\sigma|) &\Rightarrow
        (\sigma(\lambda i.t_i))(\lambda v.u_v) \coloneqq \sigma(\lambda i.t_i(\lambda v.u_v)).
    \end{align}
    Given a theory $\mathbb{T}$, we define $\mathbb{T}${\it -equivalence} as the smallest family of
    substitution-congruence $\equiv_\mathbb{T}$ on the sets $\Sigma_{\mathbb{T}}(V)$ such that
    $t \equiv_\mathbb{T} u$ for all equations $t=u \in E_{\mathbb{T}}$.
    The set $T(V)$ of $\mathbb{T}${\it -terms with variables in} $V$ is
    the quotient $\Sigma(V)/\equiv_\mathbb{T}$.
\end{definition}
When writing $\sigma(\lambda i.t_i)$ for a symbol $\sigma$, we assume that the variable $i$ ranges over $|\sigma|$.

When we see a theory $\mathbb{T}$ as specifying a computational effect as advocated in \cite{inproceedings},
$T(V)$ is seen as the set of computations with effects from $\mathbb{T}$ returning a value in $V$.
Well-known examples are theories for effects of output, state, exception, nondeterminism, and so on.
In this article, we are mainly interested in the theory of input and its expansions
(in short, we basically consider free theories).

\begin{example}
    \label{Input theory}
    \cite{garnerStreamProcessorsComodels2021}
    Given a set $A$, the theory $\mathbb{T}_A$ of $A${\it -valued input} comprises
    a single $A$-ary operation symbol $\textbf{read}_A$ and no equations.
    The set of terms $T_A(V)$ is the initial algebra $\mu X.V+X^A$.
    Its elements may be seen as $A$-ary branching trees with leaves labelled in $V$;
    or, from another perspective, they can be seen as programs
    which request $A$-values and use them to determine a return value in $V$.
\end{example}

\begin{example}
    \label{free theory}
    Consider a free theory $\mathbb{T}$ with $n$ operation symbols $\sigma_1,\ldots,\sigma_n$.
    For all $i$, we write $|\sigma_i|=A_i$.
    ($\mathbb{T}_A$ is the case $n = 1$ and $|\sigma_1|=A$.)
    The set of terms $T(V)$ is the initial algebra
    $\mu X.V+\coprod_{\sigma \in \Sigma_\mathbb{T}}X^{|\sigma|}$,
    where $\coprod$ denotes coproduct or direct sum.
    Its elements can be thought of as trees such that
    each node is labelled by a symbol $\sigma_i$ and such a node has $A_i$-ary branches
    and finally, their leaves are labelled in $V$ (or nullary operation symbols if they exist).
    Computationally, they are programs which request $n$-sorted values and return a value in $V$ depending on inputs.

    In particular, we write $\mathbb{T}^{(n)}_A$
    for the free theory with $n$ symbols and $|\sigma_1|,\ldots,|\sigma_n|$ are all the same set $A$.
    As trees, elements of $T^{(n)}_A(V)$ has the same form as elements of $T_A(V)$ but their nodes have labels in $\{1,\ldots,n\}$.
    Differences between $\mathbb{T}^{(n)}_A$ for $n \geq 2$ and $\mathbb{T}_A$ will
    become more significant when we consider their {\it comodels}.
\end{example}

\begin{definition}
    \label{Sigma structure}
    Let $\Sigma$ be a signature and $\mathscr{C}$ be a category with powers.
    A $\Sigma${\it -structure} $\bm{X}$ in $\mathscr{C}$ is
    an object $X \in \mathscr{C}$ with an {\it operation} $\sem{\sigma}_{\bm{X}} : X^{|\sigma|} \rightarrow X$ for each $\sigma \in \Sigma$.
    For each $t \in \Sigma(V)$ the {\it derived operation} $\sem{t}_{\bm{X}} : X^V \rightarrow X$ is
    determined by the recursive clauses:
    \[
    \sem{v}_{\bm{X}}=\pi_v
    \quad and \quad
    \sem{\sigma(\lambda i.t_i)}_{\bm{X}}=X^V \xrightarrow{(\sem{t_i}_{\bm{X}})_{i\in|\sigma|}} X^{|\sigma|} \xrightarrow{\sem{\sigma}_{\bm{X}}}X.
    \]
\end{definition}

\begin{definition}
    Given a theory $\mathbb{T}$, a $\mathbb{T}$-{\it{model}} in $\mathscr{C}$ is a $\Sigma$-structure $\bm{X}$ which satisfies
    $\sem{t}_{\bm{X}} = \sem{u}_{\bm{X}}$ for all equations $t=u$ of $\mathbb{T}$.
    $\mathbb{T}$-models in $\mathscr{C}$ form a category with morphisms
    $f:X \rightarrow Y$ in $\mathscr{C}$ such that the following diagram commutes for all $\sigma \in \Sigma$:
    \[
        \begin{tikzcd}
            X^{|\sigma|} \ar[r, "\sem{\sigma}_{\bm{X}}"] \ar[d, swap, "f^{|\sigma|}"] & X \ar[d, "f"] \\
            Y^{|\sigma|} \ar[r, "\sem{\sigma}_{\bm{Y}}"] & Y
        \end{tikzcd}
    \]

    The unqualified term ``model'' will mean ``model in $\mathword{Set}$''.
    We write  $\mathword{Mod}(\mathbb{T}, \mathscr{C})$ for the category of $\mathbb{T}$-models in $\mathscr{C}$,
    and $\mathword{Mod}(\mathbb{T})$ for $\mathword{Mod}(\mathbb{T},\mathword{Set})$.
\end{definition}

The set of terms $T(V)$ has a $\mathbb{T}$-model structure given by substitution.
This structure has the following universal property.

\begin{lemma}
    \label{free model}
    The set of terms $T(V)$ is the free $\mathbb{T}$-model on $V$ by
    the inclusion of variables $\eta_V:V \rightarrow T(V)$.
    That is, for any $\mathbb{T}$-model $\bm{X}$ and
    any function $f:V \rightarrow X$ to the underlying set of $\bm{X}$,
    there is the unique $\mathbb{T}$-model morphism
    $f^{\dagger}:\bm{T}(V) \rightarrow \bm{X}$ with $f^{\dagger} \circ \eta_V = f$.
    Spelling out the detail, we have $f^{\dagger}(t)=\sem{t}_{\bm{X}}(\lambda v.f(v))$.
\end{lemma}

This lemma allows us to define the {\it Kleisli category} of $\mathbb{T}$.

\begin{definition}
    \label{Kleisli category}
    For an algebraic theory $\mathbb{T}$,
    the {\it Kleisli category} $\mathope{Kl}(\mathbb{T})$ of $\mathbb{T}$
    has sets as objects.
    For sets $A, B$, the hom-set $\mathope{Kl}(\mathbb{T})(A,B)$ is defined as $\mathword{Set}(A,T(B))$.
    The identity at $A$ is $\eta_A:A \rightarrow T(A)$.
    Composition of $f:A \rightarrow T(B)$ and $g:B \rightarrow T(C)$ is
    $g^{\dagger} \circ f$ with $g^{\dagger}$ as in Lemma \ref{free model}.
\end{definition}

There are well-known functors related to $\mathope{Kl}(\mathbb{T})$;
the free functor $F_{\mathbb{T}}:\mathword{Set} \rightarrow \mathope{Kl}(\mathbb{T})$ is the identity on objects and
sends $f \in \mathword{Set}(X,Y)$ to $\eta_Y \circ f \in \mathope{Kl}(\mathbb{T})(X,Y)$.
the comparison functor $I_{\mathbb{T}}:\mathope{Kl}(\mathbb{T}) \rightarrow \mathword{Mod}(\mathbb{T})$ acts as
$A \mapsto \bm{T}(A)$ and $f \mapsto f^{\dagger}$.

We now turn to $comodel$ which is the dual notion of model.

\begin{definition}
    Let $\mathbb{T}$ be a theory. A $\mathbb{T}$-{\it{comodel}} $\bm{S}$ in a category $\mathscr{C}$ with copowers
    is a model of $\mathbb{T}$ in $\mathscr{C}^{op}$,
    i.e. an object $S \in \mathscr{C}$ with co-operations
    $\sem{\sigma}^{\bm{S}}:S \rightarrow |\sigma| \cdot S$ satisfying the equations of $\mathbb{T}$.
    Morphisms between comodels $\bm{S}$, $\bm{S'}$ are morphisms $f:S \rightarrow S'$ in $\mathscr{C}$
    such that the following diagram commutes for each symbol $\sigma$ in $\Sigma_{\mathbb{T}}$:
    \[
        \begin{tikzcd}
            S \ar[r, "\sem{\sigma}^{\bm{S}}"] \ar[d, swap, "f"] & {|\sigma| \cdot S} \ar[d, "|\sigma| \cdot f"] \\
            S' \ar[r, "\sem{\sigma}^{\bm{S'}}"] & {|\sigma| \cdot S'}
        \end{tikzcd}
    \]

    The unqualified term ``comodel'' will mean ``comodel in $\mathword{Set}$''.
    We write $\mathword{Comod}(\mathbb{T},\mathscr{C})$ for the category of $\mathbb{T}$-comodels in $\mathscr{C}$,
    and $\mathword{Comod}(\mathbb{T})$ for $\mathword{Comod}(\mathbb{T},\mathword{Set})$.
    We here note that $\mathword{Comod}(\mathbb{T},\mathscr{C}) \cong \mathword{Mod}(\mathbb{T},\mathscr{C}^{op})^{op}$,
    the opposite of the category of $\mathbb{T}$-models in the opposite of $\mathscr{C}$.
\end{definition}

\begin{example}
    A comodel $\bm{S}$ of the theory $\mathbb{T}_A$ of $A$-valued input is
    a set $S$ with a map $\sem{{\bf{read}}_A}^{\bm{S}}:S \rightarrow A \times S$.
    This map can be decomposed into two maps:
    the output map $o^{\bm{S}}:S \rightarrow A$ and
    the transition map $\partial^{\bm{S}}:S \rightarrow S$.

    We usually call elements of comodels {\it states}.
    Then a comodel $\bm{S}$ of $\mathbb{T}_A$ can be seen as a {\it state machine} which answers to
    requests for $A$-value and transition to the next state,
    determined by $o^{\bm{S}}$ and $\partial^{\bm{S}}$, respectively.
\end{example}

As explained in \cite{powerComodelsCoalgebrasState2004,10.1016/j.entcs.2008.10.018,uustaluStatefulRunnersEffectful2015},
when a theory $\mathbb{T}$ presents a computational effect, its comodels provide
deterministic environments for evaluating computations with effects from $\mathbb{T}$.
In general, given a $\mathbb{T}$-comodel $\bm{S}$ and a term $t \in T(V)$,
we have {\it derived co-operation} $\sem{t}^{\bm{S}}:S \rightarrow V \times S$ as the dual of
derived operation in Definition \ref{Sigma structure}:
\begin{align}
    \sem{v}^{\bm{S}}(s)&=\iota_v(s)=(v,s)\\
    \sem{\sigma(\lambda i.t_i)}^{\bm{S}}(s)&=\sem{t_i}^{\bm{S}}(s')
    \text{ where } \sem{\sigma}^{\bm{S}}(s)=(i,s').
\end{align}
Intuitively, evaluating a term (or a tree) $t$ with an initial state $s$
is selection of a path to a value in $t$ depending on the behavior of $s$ as follows:
first, if $t=\sigma(\lambda i.t_i)$ and $\sem{\sigma}^{\bm{S}}(s)=(i,s')$,
$s$ chooses the $i$-th branch
and transition to the next state $s'$;
then continue the computation by evaluating $t_i$ with $s'$;
finally, if the term under the chosen branch is a value,
then the evaluation terminates and returns that value.
Clearly, return values of a term $t$ appear in $t$ as variables.

We focus on the final comodel of a theory.
The final comodel of $\mathbb{T}$ is the final object of $\mathword{Comod}(\mathbb{T})$.
We will also consider the final comodel in a category $\mathscr{C}$ other than $\mathword{Set}$,
that is, the final object of $\mathword{Comod}(\mathbb{T},\mathscr{C})$.

\begin{example}
    The final comodel of $\mathbb{T}_A$ is $\bm{A}^{\Nat}$, the set of infinite sequences of elements in $A$.
    Its co-operation $\sem{\textbf{read}}^{\bm{A}^{\Nat}}:A^{\Nat} \rightarrow A \times A^{\Nat}$ is composed of
    \[
        o^{A^{\Nat}}:a_0a_1a_2\cdots \mapsto a_0
        \quad
        \partial^{A^{\Nat}}:a_0a_1a_2\cdots \mapsto a_1a_2\cdots.
    \]

    In order to help understanding the next example,
    we see a sequence in $A^{\Nat}$ in a slightly different way.
    Firstly, we see a sequence $\overrightarrow{a} \in A^{\Nat}$
    as the function $a:\Nat \rightarrow A$ with $a(k)=a_k$.
    We can get the $k$-th element $a(k)$ by
    applying $\sem{\textbf{read}}^{\bm{A}^{\Nat}}$
    (more precisely, applying its component $\partial^{A^{\Nat}}$) {$k$-times}
    and taking the head element of the resulting sequence.
    So, it is reasonable to see the domain of the function $a$ as
    $\{{\textbf{read}}\}^{\ast}=\{\epsilon,{\textbf{read}},{\textbf{read}} \cdot {\textbf{read}},\ldots\}$,
    the set of finite repeats of the symbol ${\textbf{read}}$.
    Then components of $\sem{\textbf{read}}^{\bm{A}^{\Nat}}$ can be written as:
    \[
        o^{A^{\Nat}}:a \mapsto a(\epsilon)
        \quad
        \partial^{A^{\Nat}}:a \mapsto a({\textbf{read}} \_)
    \]
    where $(a({\textbf{read}} \_))({\textbf{read}}^k)=a({\textbf{read}} \cdot {\textbf{read}}^k)=a({\textbf{read}}^{k+1})$.
\end{example}

We recognize comodels and the final comodel of a free theory
in a similar way to the case of the theory of input.

\begin{example}
    \label{final model of free theory}
    Let $\mathbb{T}$ be a free theory.
    A comodel of $\mathbb{T}$ is
    a state machine that answers to requests for elements of $|\sigma|$ for each $\sigma \in \Sigma_{\mathbb{T}}$
    , and then transition to a next state depending on the requested operation symbol.
    This comprises a set of states $S$ and operations
    $\sem{\sigma}^{\bm{S}}=(o^{\bm{S}}_{\sigma}, \partial^{\bm{S}}_{\sigma}):S\rightarrow |\sigma|\times S$
    giving for each state $s \in S$ an output $o^{\bm{S}}_{\sigma} (s) \in |\sigma|$
    and a next state $\partial^{\bm{S}}_{\sigma} (s) \in S$.
    By taking the product of $\sem{\sigma}^{\bm{S}}$'s,
    we can regard the comodel $\bm{S}$ as a coalgebra of the functor
    $H_{\mathbb{T}} \coloneqq \prod_{\sigma \in \Sigma_{\mathbb{T}}} (|\sigma| \times (\_)) \cong (\prod_{\sigma \in \Sigma_{\mathbb{T}}} |\sigma|) \times (\_)^{\Sigma_{\mathbb{T}}}$.
    Moreover, morphisms between $\mathbb{T}$-comodels coincide with those between $H_{\mathbb{T}}$-coalgebras.
    Therefore the final $\mathbb{T}$-comodel $\bm{S}_{\mathbb{T}}$ is
    the final $H_{\mathbb{T}}$-coalgebra $\nu X.(\prod_{\sigma \in \Sigma_{\mathbb{T}}} |\sigma|) \times X^{\Sigma_{\mathbb{T}}}=$
    $(\prod_{\sigma \in \Sigma_{\mathbb{T}}} |\sigma|)^{{\Sigma_{\mathbb{T}}}^{\ast}}$.
    The comodel structure of $\bm{S}_{\mathbb{T}}$ is given by
    \[
        o_{\sigma}^{\bm{S}_{\mathbb{T}}}(s)=\pi_{\sigma}(s(\epsilon))
        \quad
        \partial_{\sigma}^{\bm{S}_{\mathbb{T}}}(s)=s(\sigma \_)
    \]
    where $\pi_{\sigma}$ denotes the projection map
    $(\prod_{\sigma \in \Sigma_{\mathbb{T}}} |\sigma|) \rightarrow |\sigma|$.
\end{example}

We omit indices $\bm{X},\ \bm{S}$ of $\sem{\sigma}_{\bm{X}}$, $\sem{\sigma}^{\bm{S}}$,
$o^{\bm{S}}_{\sigma}$ and $\partial^{\bm{S}}_{\sigma}$, if it is clear from contexts.

An important property of the final comodel is that
it describes ``observable behaviors'' of states in comodels.
As described in \cite{garnerCostructurecosemanticsAdjunctionComodels2020},
the final comodel of the theory $\mathbb{T}$ can be characterized as
the set of all possible states of all possible comodels modulo {\it operational equivalence}.

\begin{definition}
    Let $\mathbb{T}$ be an algebraic theory.
    For states $s_1 \in \bm{S}_1$, $s_2 \in \bm{S}_2$ of two $\mathbb{T}$-comodels,
    we say that they are {\it operationally equivalent} if
    for all $\mathbb{T}$-terms $t$, $\pi_V(\sem{t}^{\bm{S}_1}(s_1))=\pi_V(\sem{t}^{\bm{S}_2}(s_2))$.
\end{definition}

\begin{lemma}[\cite{garnerCostructurecosemanticsAdjunctionComodels2020}]
    States $s_1 \in \bm{S}_1$ and $s_2 \in \bm{S}_2$ of two $\mathbb{T}$-comodels are
    operationally equivalent iff they become equal under each unique map to
    the final $\mathbb{T}$-comodel.
\end{lemma}

\section{Operational Sub-basis on Comodels}
\label{sec:3}

Garner \cite{garnerStreamProcessorsComodels2021} defined the {\it operational topology} on a comodel of an algebraic theory $\mathbb{T}$ as
the topology whose basic open sets describe those states which are
indistinguishable with respect to a finite set of $\mathbb{T}$-computations.

\begin{definition}
    \cite{garnerStreamProcessorsComodels2021}
    Let $\bm{S}$ be a $\mathbb{T}$-comodel.
    The {\it operational topology} on $S$ is generated by sub-basic open sets
    \[
        [t \mapsto v]_{\bm{S}} \coloneqq \{s\in S \mid \sem{t}^{\bm{S}}(s)=(v,s') \text{ for some } s'\}
        \qquad
        \text{for all } t \in T(V) \text{ and } v \in V.
    \]
\end{definition}

We omit the subscript $\bm{S}$ of $[t \mapsto v]_{\bm{S}}$ if the comodel is clear.

We call the sub-basis $\{[t \mapsto v]_{\bm{S}} \mid t:\text{term, } v:\text{variable}\}$
of the operational topology, the {\it operational sub-basis} of the comodel $\bm{S}$.
In particular, for the final comodel of a theory $\mathbb{T}$,
we write $\Phi_{\mathbb{T}}$ for its operational sub-basis.
In the sequel, we mainly use this operational sub-basis, not the operational topology.
Thus we define a category whose objects are sets with sub-basis.

\begin{definition}
    We define the category $\cat$ of {\it sets with sub-basis and functions continuous on sub-basis}.
    Objects are pairs $(X,\Phi)$, $X$ is a set and $\Phi$ is a subset of $\mathcal{P}(X)$
    which contains $\emptyset$ and $X$.
    A morphism $(X,\Phi) \rightarrow (Y,\Psi)$ is a mapping $f:X \rightarrow Y$ such that
    for each set $U \in \Psi$, $f^{-1}(U) \in \Phi$. We call such maps as {\it continuous functions on sub-basis}.
\end{definition}

We can show that $\cat$ is a cocomplete category and its constructions of colimits are similar to colimits in $\mathword{Top}$.
As for comodels in $\cat$, the following is valid because we can show that,
if a set with sub-basis $(X,\Phi)$ has a $\mathbb{T}$-comodel structure in $\cat$,
then the sub-basis $\Phi$ contains all sets of the form $[t \mapsto v]$.
\begin{proposition}
    \label{prop of adjunction}
    There is an adjunction
    \[
        \begin{tikzpicture}[auto]
            \node (C) at (0,0) {$\mathword{Comod}(\mathbb{T},\cat)$}; \node (D) at (3,0) {$\mathword{Comod}(\mathbb{T})$};
            \draw[<-, transform canvas={yshift=6pt}] (C) --node {$\scriptstyle op$}(D);
            \draw[->, transform canvas={yshift=-6pt}] (C) --node[swap, label=above:$\top$] {$\scriptstyle U$} (D);
        \end{tikzpicture}
    \]
    where $op$ gives a comodel $\bm{S}$ the operational sub-basis, and $U$ forgets sub-basis.
\end{proposition}

As a corollary, the functor $op$ preserves the final object
(A similar result for the category $\mathword{Top}$ is in \cite{garnerStreamProcessorsComodels2021}).

\begin{corollary}
    \label{cor:1}
    For each theory $\mathbb{T}$, the final $\mathbb{T}$-comodel in $\mathword{Set}$ is, when endowed with
    its operational sub-basis, the final comodel in $\cat$.
\end{corollary}

\section{Residual Comodel and Continuous Functions}
\label{sec:4}

{\it Residual comodels} allow us to describe stateful translations between different notions of computation.
The word {\it residual} comes from \cite{katsumataInteractionLawsMonads2019}.

\begin{definition}
    \cite{garnerStreamProcessorsComodels2021}
    Let $\mathbb{T}$ and $\mathbb{T'}$ be theories.
    An {\it $\mathbb{T}$-residual $\mathbb{T'}$-comodel} is
    a comodel of $\mathbb{T'}$ in the Kleisli category $\mathope{Kl}(\mathbb{T})$.
\end{definition}

Spelling out the detail, an $\mathbb{T}$-residual $\mathbb{T'}$-comodel $\bm{S}$ comprises
an underlying set $S$ and a co-operation $\sem{\sigma}^{\bm{S}}:S \rightarrow T(|\sigma| \times S)$
for each symbol $\sigma \in \Sigma_{\mathbb{T'}}$.
That is, for each state $s \in S$ and each symbol $\sigma$,
we need to deal with an $\mathbb{T}$-computation
in order to decide what state we should transition to and
to extract an index in $|\sigma|$ from $\sem{\sigma}^{\bm{S}}(s)$.
The derived co-operation $\sem{t}^{\bm{S}}:S \rightarrow T(V \times S)$ of a term $t \in T'(V)$ is
calculated using composition in the Kleisli category:
\begin{align}
    &\sem{v}^{\bm{S}}(s)=(v,s) \in V \times S \subseteq T(V \times S)\\
    &\sem{\sigma(\lambda i.t_i)}^{\bm{S}}(s)=
    \sem{\sigma}^{\bm{S}}(s)(\lambda (i,s').\sem{t_i}^{\bm{S}}(s')),
\end{align}
where the term $\sem{\sigma}^{\bm{S}}(s)(\lambda (i,s').\sem{t_i}^{\bm{S}}(s')) \in T(V \times S)$ is
the substitution of {$(\sem{t_i}^{\bm{S}}(s'))_{(i,s')} \in T(V \times S)^{|\sigma| \times S}$}
to $\sem{\sigma}^{\bm{S}}(s) \in T(|\sigma| \times S)$.

By the above intuition about $\mathbb{T}$-residual $\mathbb{T'}$-comodels,
we expect that when we have a state $s$ of an $\mathbb{T}$-residual $\mathbb{T'}$-comodel $\bm{S}$ and
a state $m$ of an $\mathbb{T}$-comodel,
then, for each term $t \in T'(V)$,
we can evaluate $\sem{t}^{\bm{S}}(s) \in T(V \times S)$ with the initial state $m$
to a value-state pair in $V \times S$.
This idea is formalized as:

\begin{definition}
    \label{tensor in set}
    \cite{garnerStreamProcessorsComodels2021}
    Let $\mathbb{T}$, $\mathbb{T'}$ be theories.
    Let $\bm{S}$ be an $\mathbb{T}$-residual $\mathbb{T'}$-comodel,
    and let $\bm{M}$ be an $\mathbb{T}$-comodel.
    The {\it tensor product} $\bm{S}\cdot\bm{M}$ is the $\mathbb{T'}$-comodel
    with underlying set $S \times M$ and co-operations
    \[
        \sem{\sigma}^{\bm{S} \cdot \bm{M}}:
        S \times M \xrightarrow{\sem{\sigma}^{\bm{S}} \times M}
        T(|\sigma| \times S) \times M
        \xrightarrow{(t,m) \mapsto \sem{t}^{\bm{M}}(m)}
        |\sigma| \times S \times M.
    \]
\end{definition}

This construction can be generalized to the case that
$\bm{M}$ is a $\mathbb{T}$-comodel in a category $\mathscr{C}$ with copowers.
For an $\mathbb{T}$-residual $\mathbb{T'}$-comodel $\bm{S}$ and an $\mathbb{T}$-comodel $\bm{M}$ in $\mathscr{C}$,
there is a $\mathbb{T'}$-comodel in $\mathscr{C}$ whose underlying object is the copower $S \cdot M$.

If there is the final $\mathbb{T'}$-comodel $\bm{S}_{\mathbb{T'}}$ in $\mathscr{C}$,
$\mathword{Comod}(\mathbb{T'},\mathscr{C})(\bm{S \cdot M},\bm{S}_{\mathbb{T'}})$ has only one map $e$.
Now we consider the case of $\mathscr{C} = \cat$ and $\bm{M}$ is the final $\mathbb{T}$-comodel $\bm{S}_{\mathbb{T}}$.
By Corollary \ref{cor:1}, we have this $e:S \times S_{\mathbb{T}} \rightarrow S_{\mathbb{T'}}$ called the extent of $\bm{S}$.
Then its currying $\lambda s \in S.e(s,\_):S \rightarrow \cat(S_{\mathbb{T}},S_{\mathbb{T'}})$
translates elements of a residual comodel to continuous functions between two final comodels,
in particular, these functions are continuous on sub-basis.

One of our goals is to show that, for free theories $\mathbb{T}$ and $\mathbb{T'}$,
the currying of the extent of the {\it final $\mathbb{T}$-residual $\mathbb{T'}$-comodel} is surjective.
To define the final object, we first define the notion of morphism between residual comodels.
This is different from the usual morphism between comodels in the Kleisli category.

\begin{definition}
    \cite{garnerStreamProcessorsComodels2021}
    Let $\mathbb{T}$ and $\mathbb{T'}$ be theories,
    $\bm{S}$ and $\bm{U}$ be $\mathbb{T}$-residual $\mathbb{T'}$-comodels.
    {\it A map of residual comodels} $\bm{S} \rightarrow \bm{U}$ is
    a function $f:S \rightarrow U$ such that the following diagram commutes
    for each symbol $\sigma \in \Sigma_{\mathbb{T'}}$:
    \[
        \begin{tikzcd}
            S \ar[r, "\sem{\sigma}^{\bm{S}}"] \ar[d, swap, "f"] &
            {T(|\sigma| \times S)} \ar[d, "T(|\sigma| \times f)"]\\
            U \ar[r, swap, "\sem{\sigma}^{\bm{U}}"] &
            {T(|\sigma| \times U)}
        \end{tikzcd}
    \]
\end{definition}

Now we define the final $\mathbb{T}$-residual $\mathbb{T'}$-comodel $\bm{I}_{\mathbb{T},\mathbb{T'}}$
as the final object of the category of residual comodels and maps between residual comodels,
which is different from the final $\mathbb{T'}$-comodel in $\mathope{Kl}(\mathbb{T})$.

We call the currying of the extent of $\bm{I}_{\mathbb{T},\mathbb{T'}}$ the {\it reflection}.

\begin{definition}
    Let $\mathbb{T}$ and $\mathbb{T'}$ be free theories,
    $\bm{I}_{\mathbb{T},\mathbb{T'}}$ be the final $\mathbb{T}$-residual $\mathbb{T'}$-comodel.
    The {\it reflection} function is defined as currying of the extent $e$ of
    $\bm{I}_{\mathbb{T},\mathbb{T'}}$:
    \[
        {\textsf{reflect}}:I_{\mathbb{T},\mathbb{T'}} \rightarrow \cat(S_{\mathbb{T}},S_{\mathbb{T'}}) \qquad
        s \mapsto e(s,\_):S_{\mathbb{T}} \rightarrow S_{\mathbb{T'}}.
    \]
\end{definition}

Following examples explain how a state of residual comodel implements a function
and why we restrict our target to maps in $\cat$ (or straight functions in the introduction).

\begin{example}
    Let $A=\{0,1\}$ and $B = \{a,b\}$.
    Consider a function $A^{\Nat} \rightarrow B^{\Nat}$ which rewrites $0$ to $a$ and $1$ to $b$.
    We take theories $\mathbb{T}$ and $\mathbb{T'}$ as $\mathbb{T}_A$ and $\mathbb{T}_B$.
    The final $\mathbb{T}$-residual $\mathbb{T'}$-comodel $I_{\mathbb{T},\mathbb{T'}}$ has
    its residual comodel structure $\sem{\textbf{read}_B}:I_{\mathbb{T},\mathbb{T'}} \rightarrow T(B \times I_{\mathbb{T},\mathbb{T'}})$.
    Take a state $s \in I_{\mathbb{T},\mathbb{T'}}$ as
    \begin{equation}
        \sem{\textbf{read}_B}(s)=
        \begin{tikzpicture}[auto,baseline=-20pt]
            \node (root) at (0,0) {${\textbf{read}}_A$};
            \node (-0) at (-1,-1.5) {$(a,s)$};
            \node (-1) at (1,-1.5) {$(b,s)$};
            \draw (root) --node[swap]{$0$} (-0);
            \draw (root) --node{$1$} (-1);
        \end{tikzpicture}
    \end{equation}
    Then for a given stream, such as $1011 \cdots$ in $A^{\Nat}$, this state constructs a stream $y_1y_2 \cdots$ in $B^{\Nat}$ as follows:
    To compute the first digit $y_1$, it uses $\textbf{read}_B$ one time.
    The tree $\sem{\textbf{read}_B}(s)$ requires reading an $A$-element and this requirement is met by
    the given stream.
    So it consumes the first digit $1$ of the input and it determines that $y_1=b$.
    To compute $y_2$, it uses $\textbf{read}_B$ twice.
    The first $\textbf{read}_B$ is computed as above and it reaches the leaf $(b,s)$.
    The second one is applied to this new $s$ and now it consumes the second digit $0$ of the input.
    Thus it reaches $(a,s)$ and $y_2$ becomes $a$.
    The computation continues similarly and it will implement the function considered as above.
    We can implement more complex maps by deepening the tree or by using other states in leaves.
\end{example}

\begin{example}
    This example exhibits a function which cannot be implemented by states of residual comodels.
    Let $\mathbb{T}=\mathbb{T}^2_{\Nat}$ and $\mathbb{T'}=\mathbb{T}_{\Nat}$.
    Consider a function between final comodels $S_{\mathbb{T}} \rightarrow S_{\mathbb{T'}}$
    which takes the sum of each depth:
    
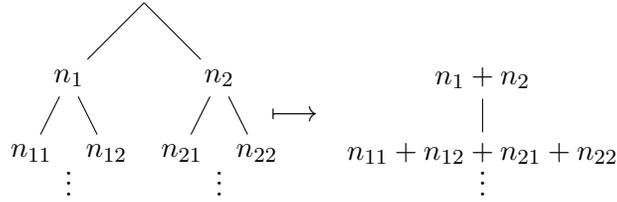
\begin{equation}
        \begin{tikzpicture}[auto, baseline=-15pt]
            \node (0) at (-1,-1) {$n_1$};
            \node (1) at (1,-1) {$n_2$};
            \node (00) at (-1.5,-2) {$n_{11}$};
            \node (01) at (-0.5,-2) {$n_{12}$};
            \node (10) at (0.5,-2) {$n_{21}$};
            \node (11) at (1.5,-2) {$n_{22}$};
            \node at (-1,-2.3) {$\vdots$};
            \node at (1,-2.3) {$\vdots$};
            \draw (0,0) -- (0); \draw (0,0) -- (1);
            \draw (0) -- (00); \draw (0) -- (01);
            \draw (1) -- (10); \draw (1) -- (11);

            \node at (2,-1.5) {$\longmapsto$};

            \node (s1) at (4.5,-1) {$n_1+n_2$};
            \node (s2) at (4.5,-2) {$n_{11}+n_{12}+n_{21}+n_{22}$};
            \node at (4.5,-2.3) {$\vdots$};
            \draw (s1) -- (s2);
        \end{tikzpicture}
    \end{equation}
    If this can be implemented by a state $s \in I_{\mathbb{T},\mathbb{T'}}$,
    the first digit $n_1+n_2$ of the output is computed by the term $\sem{\textbf{read}_{\Nat}}(s) \in T^{(2)}_{\Nat}(\Nat \times I_{\mathbb{T},\mathbb{T'}})$.
    There are three cases; (i) $\sem{\textbf{read}_{\Nat}}(s)$ is a variable $(n,s') \in \Nat \times I_{\mathbb{T},\mathbb{T'}}$,
    (ii) $\sem{\textbf{read}_{\Nat}}(s)$ is of the form $\textbf{read}^1_{\Nat}(\lambda n.u_n)$ and
    (iii) $\sem{\textbf{read}_{\Nat}}(s)$ is of the form $\textbf{read}^2_{\Nat}(\lambda n.u_n)$.
    When (i), its output is always $n$ and thus this cannot depend on $n_1,n_2$.
    When (ii), it reads $n_1$ of the input tree, selects the term $u_{n_1}$
    and, if $u_{n_1}$ requires further input, it uses the tree under $n_1$.
    So, in this case, the output cannot depend on $n_2$.
    Similarly, when (iii), the output cannot depend on $n_1$.
    Consequently, the term $\sem{\textbf{read}_{\Nat}}(s)$ cannot observe both of $n_1$ and $n_2$
    and thus the function summing up each depth cannot be implemented by residual comodels.
\end{example}

We will show the surjectivity of the reflection by characterizing $\cat(S_{\mathbb{T}},S_{\mathbb{T'}})$ as the {\it final $\mathbb{T}$-$\mathbb{T'}$-bimodel}.
Here, for theories $\mathbb{T}$ and $\mathbb{T'}$,
the category of  $\mathbb{T}$-$\mathbb{T'}$-bimodels is the category of $\mathbb{T'}$-comodels in $\mathword{Mod}(\mathbb{T})$, i.e., $\mathword{Comod}(\mathbb{T'},\mathword{Mod}(\mathbb{T}))$.
We only describe the most important properties for our purpose; bimodels can be seen as residual comodels.

\begin{lemma}
    \label{lem:bimodel}
    Let $\mathbb{T}$ be any theory.
    For any $\mathbb{T}$-model $\bm{X}=(X,\sem{\_}_{\bm{X}})$ and set $B$,
    the copower $B \cdot \bm{X}$ is the quotient of the free model $\bm{T}(B \times X)$
    for an $\mathbb{T}$-congruence relation.

    Especially, if $\mathbb{T}$ is a free theory,
    we can take a canonical representative of each equivalence class
    and thus $B \cdot X$ may be regarded as a subset of $T(B \times X)$.
    In detail, the set of canonical representatives coincides with the set of terms in $T(B \times X)$
    which have no non-trivial sub-terms whose variables are labelled by the same element of $B$,
    in other words, banned sub-terms are of the form $\sigma(\lambda i.(b,x_i))$ for $b \in B$.
    In this case, the $\mathbb{T}$-model structure of $B \cdot \bm{X}$ is that of $T(B \times X)$
    except that $\sem{\sigma}_{B \cdot \bm{X}}(\lambda i.(b,x_i))=(b,\sem{\sigma}_{\bm{X}}(\lambda i.x_i))$.
\end{lemma}

\begin{proof}
    Define a $\mathbb{T}$-congruence $\sim$ on $\bm{T}(B \times X)$ as
    the minimal congruence satisfying
    \begin{equation}
        \label{eq:6}
        \sigma(\lambda i.(b,x_i)) \sim (b,\sem{\sigma}_{\bm{X}}(\lambda i.x_i)).
    \end{equation}
    for all symbols $\sigma \in \Sigma_{\mathbb{T}}$.
    The quotient $\bm{T}(B \times X) / \sim$ satisfies universality of
    the copower $B \cdot \bm{X}$.

    If $\mathbb{T}$ is a free theory,
    by orienting read \eqref{eq:6} from left to right,
    this determines a strongly normalizing rewrite system on $T(B \times X)$;
    if there is a sub-term of the form $\sigma(\lambda i.(b,x_i))$ then
    rewrite this into the variable $(b,\sem{\sigma}_{\bm{X}}(\lambda i.x_i))$.
    Thus we can take the normal forms as representatives of equivalence classes.
\end{proof}

\begin{remark}
    The latter of this lemma is justified
    because the set of $\mathbb{T}$-terms $T(B \times X)$ coincides with
    the set of $\Sigma_{\mathbb{T}}$-terms $\Sigma_{\mathbb{T}}(B \times X)$,
    whose elements are trees with $\Sigma_{\mathbb{T}}$-labelled nodes and
    $(B \times X)$-labelled leaves.
    When $\mathbb{T}$ has non-trivial equations,
    $T(B \times X)$ is a quotient of $\Sigma_{\mathbb{T}}(B \times X)$.
    Thus we must argue about rewriting systems on a quotient set and
    cannot generalize this lemma simply.
\end{remark}

\begin{definition}
    \label{def:1}
    Let $\mathbb{T}$ and $\mathbb{T'}$ be free theories
    and $\bm{K}$ be a $\mathbb{T}$-$\mathbb{T'}$-bimodel.
    We define the $\mathbb{T}$-residual $\mathbb{T'}$-comodel
    $\bm{\check K}=(K,\sem{\_}^{\bm{\check K}})$ whose co-operations are
    the composites
    \[
        \sem{\sigma}^{\bm{\check K}}: K \xrightarrow{\sem{\sigma}^{\bm{K}}} |\sigma| \cdot K
        \hookrightarrow T(|\sigma| \times K)
    \]
    of the $\mathbb{T'}$-comodel structure map with the inclusion of Lemma \ref{lem:bimodel}.
\end{definition}

\begin{definition}
    Let $\mathbb{T}$ and $\mathbb{T'}$ be free theories
    and $\bm{E}_{\mathbb{T},\mathbb{T'}}$ be the final $\mathbb{T}$-$\mathbb{T'}$-bimodel.
    Then we have the unique $\mathbb{T}$-residual $\mathbb{T'}$-comodel map
    $\bm{\check{E}}_{\mathbb{T},\mathbb{T'}} \rightarrow \bm{I}_{\mathbb{T},\mathbb{T'}}$
    and we define the {\it reification} function as its underlying map:
    \[{\textsf{reify}}:E_{\mathbb{T},\mathbb{T}} \rightarrow I_{\mathbb{T},\mathbb{T'}}.\]
\end{definition}

If we characterize $\cat(S_{\mathbb{T}},S_{\mathbb{T'}})$ as the final $\mathbb{T}$-$\mathbb{T'}$-bimodel
and if we show  $\textsf{reflect} \circ \textsf{reify}$ is the identity map,
we can conclude that $\textsf{reflect}$ is surjective.

\section{The Final Bimodel of Free Theories}
\label{sec:6}

In this section, we assume that $\mathbb{T}$ and $\mathbb{T'}$ are free theories
(the only exception is in Proposition \ref{prop:1}).
We write $\bm{S}_{\mathbb{T}}$ and $\bm{S}_{\mathbb{T'}}$ for their final comodels
and we regard them as objects of $\cat$ with their operational sub-bases $\Phi_{\mathbb{T}}$ and $\Phi_{\mathbb{T'}}$.

Our goal in this section is to prove that $\cat(S_{\mathbb{T}},S_{\mathbb{T'}})$ is the final $\mathbb{T}$-$\mathbb{T'}$-bimodel.
This is verified by proving that the functor $\cat(\bm{S}_{\mathbb{T}},\_):\cat \rightarrow \mathword{Mod}(\mathbb{T})$
preserves the final $\mathbb{T'}$-comodel.
Here we assert that $\cat(\bm{S}_{\mathbb{T}},\_)$ is actually a functor to $\mathword{Mod}(\mathbb{T})$.
For each object $X \in \cat$,
the $\mathbb{T}$-model structure on $\cat(S_{\mathbb{T}},X)$ is
\begin{equation}
    \label{structure}
    \begin{array}{rccc}
        \sem{\sigma}=\mathope{split}_{\sigma}:&
        \cat (S_{\mathbb{T}}, X)^{|\sigma|} & \rightarrow & \cat (S_{\mathbb{T}}, X)\\
        & \rotatebox{90}{$\in$} & & \rotatebox{90}{$\in$}\\
        & (f_i)_{i \in |\sigma|} & \mapsto &
        \begin{array}{ccc}
            S_{\mathbb{T}} & \rightarrow & X\\
            \rotatebox{90}{$\in$} & & \rotatebox{90}{$\in$}\\
            s & \mapsto & f_{o_{\sigma}s}(\partial_{\sigma}s)\\
        \end{array}\\
    \end{array}
\end{equation}
for each $\sigma \in \Sigma_{\mathbb{T}}$.

The outline of the proof is as follows:
\begin{enumerate}[label={(\Roman*)}]
    \item The functor $\cat(\bm{S}_{\mathbb{T}},\_):\cat \rightarrow \mathword{Mod}(\mathbb{T})$ has a left adjoint.
    \item $\cat(\bm{S}_{\mathbb{T}},\_)$ preserves copowers of objects which have a {\it simple} (see Definition \ref{def:2}) {\it $\mathbb{T'}$-comodel} structure.
    \item The final $\mathbb{T'}$-comodel $\bm{S}_{\mathbb{T'}}$ (endowed with $\Phi_{\mathbb{T'}}$)
    is a {\it simple} $\mathbb{T'}$-comodel.
    \item Conclude the claim by using adjointness in (I).
\end{enumerate}

(I) is established in \cite{freydAlgebraValuedFunctors1966} and \cite{garnerStreamProcessorsComodels2021}.
We cite the statement from \cite{garnerStreamProcessorsComodels2021}.

\begin{proposition}
    \label{prop:1}
    Let $\mathbb{T}$ be a theory (which is not necessarily free).
    Let $\mathscr{C}$ be a category with copowers and
    $\bm{S}$ a $\mathbb{T}$-comodel in $\mathscr{C}$.
    For any object $C \in \mathscr{C}$,
    the hom-set $\mathscr{C}(S,C)$ bears a structure of $\mathbb{T}$-model
    $\mathscr{C}(\bm{S},C)$ with operations
    \begin{equation}
        \label{eq:7}
        \sem{\sigma}_{\mathscr{C}(\bm{S},C)}(\lambda i. S \xrightarrow{f_i}C)=
        S \xrightarrow{\sem{\sigma}^{\bm{S}}} |\sigma| \cdot S
        \xrightarrow{\lrangle{f_i}_{i \in |\sigma|}} C
    \end{equation}
    where $\lrangle{f_i}_{i \in |\sigma|}$ is the copairing of the $f_i$'s.
    As $C$ varies, this assignment underlies a functor
    $\mathscr{C}(\bm{S},\_) : \mathscr{C} \rightarrow \mathword{Mod}(\mathbb{T})$.
    If $\mathscr{C}$ is cocomplete, this functor has
    a left adjoint $(\_) \otimes \bm{S} : \mathword{Mod}(\mathbb{T}) \rightarrow \mathscr{C}$.
\end{proposition}

\begin{remark}
    The $\mathbb{T}$-model structure \eqref{structure} on $\cat(S_{\mathbb{T}},X)$
    is given by this proposition with the ordinary $\mathbb{T}$-comodel structure of $\bm{S}_{\mathbb{T}}$.
\end{remark}

(III) is clear from the definition of {\it simplicity} and (IV) is shown as follows:

When (I),(II) and (III) have been verified, then we have the adjunction in (I) as
\[(\_) \otimes \bm{S}_{\mathbb{T}} \dashv \cat(\bm{S}_{\mathbb{T}},\_):\cat \rightarrow \mathword{Mod}(\mathbb{T}).\]
Since the left adjoint $(\_) \otimes \bm{S}_{\mathbb{T}}$ preserves copowers and
$\cat(\bm{S}_{\mathbb{T}},\_)$ also preserves copowers of objects $X$ in $\cat$
with a {\it simple} $\mathbb{T'}$-comodel structure,
we have the following isomorphism for an arbitrary object $\bm{Y}$ in $\mathword{Comod}(\mathbb{T'},\mathword{Mod}(\mathbb{T}))$,
\[\mathword{Comod}(\mathbb{T'},\mathword{Mod}(\mathbb{T}))(\bm{Y},\cat(\bm{S}_{\mathbb{T}},\bm{X})) \cong \mathword{Comod}(\mathbb{T'},Sub)(\bm{Y} \otimes \bm{S}_{\mathbb{T}},\bm{X}).\]
The final $\mathbb{T'}$-comodel $\bm{S}_{\mathbb{T'}}$ is {\it simple}.
So we let $\bm{X}=\bm{S}_{\mathbb{T'}}$ in above, then
\[\mathword{Comod}(\mathbb{T'},\mathword{Mod}(\mathbb{T}))(\bm{Y},\cat(\bm{S}_{\mathbb{T}},\bm{S}_{\mathbb{T'}})) \cong \mathword{Comod}(\mathbb{T'},Sub)(\bm{Y} \otimes \bm{S}_{\mathbb{T}},\bm{S}_{\mathbb{T'}}).\]
By finality of $\bm{S}_{\mathbb{T'}}$, the right hand side is a singleton.
Thus, $\cat(\bm{S}_{\mathbb{T}},\bm{S}_{\mathbb{T'}})$ is final in $\mathword{Comod}(\mathbb{T'},\mathword{Mod}(\mathbb{T}))$.

\vsp

Therefore we will concentrate on (II).
First, we define the notion of {\it simple comodels} appearing in (II).

\begin{definition}
    \label{def:2}
    We say a comodel $\bm{S}$ is {\it simple} if observationally equivalent states are actually identical,
    or explicitly, if $\bm{S}$ satisfies following condition for all states $s,s'$:
    \[
        \begin{array}{l}
            \text{if}\ 
            o_{\sigma}(\partial_{\sigma_n}\cdots\partial_{\sigma_1}(s))=
            o_{\sigma}(\partial_{\sigma_n}\cdots\partial_{\sigma_1}(s'))\\
            \text{for all sequences of symbols}\ \sigma_1, \ldots, \sigma_n\ \text{and for all}\ \sigma,\\
            \text{then}\ s=s'.
        \end{array}
    \]
\end{definition}

This definition says that a comodel is simple iff it has no proper quotient
(this is the original definition of simple comodels in \cite{ruttenUniversalCoalgebraTheory2000}).
The final comodel is clearly a simple comodel (this is a justification of (III)).

The statement (II) says that,
for a $\mathbb{T'}$-comodel $((X,\Phi),\sem{\_}^{\bm{X}})$ in $\cat$ whose
underlying comodel $(X,\sem{\_}^{\bm{X}})$ is simple,
we have the copower $(\iota_i:X \rightarrow I \cdot X)_{i \in I}$ in $\cat$,
then maps given by applying the functor $\cat(\bm{S}_{\mathbb{T}},\_)$ to each $\iota_i$
\begin{equation}
    \label{maps}
    (\iota_i \circ (\_):\cat(\bm{S}_{\mathbb{T}},X) \rightarrow \cat(\bm{S}_{\mathbb{T}},I \cdot X))_{i \in I}
\end{equation}
constitute a copower cocone in $\mathword{Mod}(\mathbb{T})$.
We prove this in two steps: uniqueness of the mediating morphism and existence of it.

The difficult part is to prove its uniqueness and this is rephrased as following.
\begin{theorem}[uniqueness of the mediating morphism]
    \label{main thm}
    Let $((X,\Phi),\sem{\_}^{\bm{X}})$ be a $\mathbb{T'}$-comodel in $\cat$ whose
    underlying comodel $(X,\sem{\_}^{\bm{X}})$ is simple.
    For its copower $(\iota_i:X \rightarrow I \cdot X)_{i \in I}$ in $\cat$,
    the family of maps given by applying the functor $\cat(\bm{S}_{\mathbb{T}},\_)$ to each $\iota_i$
    \begin{equation}
        (\iota_i \circ (\_):\cat(\bm{S}_{\mathbb{T}},X) \rightarrow \cat(\bm{S}_{\mathbb{T}},I \cdot X))_{i \in I}
    \end{equation}
    is jointly epimorphic in the category $\mathword{Mod}(\mathbb{T})$
    (that is, if $\mathbb{T}$-model maps $f,g:\cat(\bm{S}_{\mathbb{T}},I \cdot X) \rightarrow \bm{Y}$ satisfy
    $f \circ (\iota_i \circ (\_))=g \circ (\iota_i \circ (\_))$ for all $i \in I$,
    then $f=g$).
\end{theorem}

Here we describe an overview of the proof:
\begin{enumerate}[label={(\roman*)}]
    \item Let $M$ be the subset of $\cat(S_{\mathbb{T}},I \cdot X)$ which is generated by
    the image of maps \eqref{maps} and
    $\mathbb{T}$-model structure maps on $\cat(S_{\mathbb{T}},I \cdot X)$ \eqref{structure}.
    \item Then it suffices to show that $M=\cat(S_{\mathbb{T}},I \cdot X)$,
    i.e., each map $f \in \cat(S_{\mathbb{T}},I \cdot X)$ can be expressed by $\mathope{split}_{\sigma}$'s and
    maps $g$ whose image $g(S_{\mathbb{T}})$ is a subset of $\iota_i(X)$ for some $i$.
    \item We show (ii) by induction on the number of indices $i \in I$ such that
    $\iota_i(X) \cap f(S_{\mathbb{T}}) \neq \emptyset$ (we define $I_f$ as the set of such indexes).
    \item When $|I_f|=1$ (in particular $|I|=1$), there is nothing to do.
    \item When $|I_f| > 1$, we will show that $f$ can be expressed by $\mathope{split}_{\sigma}$ and
    maps $g$ which satisfy $I_g \subsetneq I_f$.
    \item If $I_f$ is finite, the above argument in fact shows $f \in M$.
    \item To complete the proof, we will show that even if $I_f$ is infinite, the situation comes to the finite case.
\end{enumerate}

(We ignore the case of $|I_f|=0$ since it is equivalent to $|I|=0$.)

The key points are how to express $f$ by using $\mathope{split}_{\sigma}$ as (v) and
how to show (vii).
Firstly, argue about (v).
When the theory $\mathbb{T}$ is $\mathbb{T}_A$ (the case Garner dealt with),
this has only one symbol $\textbf{read}_A$ and the final comodel $S_{\mathbb{T}}$ is $A^{\Nat}$.
For a given map $f$, the map $\mathope{split}(\lambda a \in A. f(a \_))$ is equal to $f$:
for a given sequence $\overrightarrow{a}$,
the operation $\mathope{split}$ separates it into the head $a_0$ and the tail $a_1 a_2 \cdots$,
then $\lambda a. f(a \_)$ simply reconnects them (and apply $f$).

When the final comodel consists of trees (when $\mathbb{T}$ has more than two symbols), the situation is not so simple.
For a given $f$, we want an expression $f = \mathope{split}_{\sigma}(\lambda i \in |\sigma|.f_i)$.
There are many problems.
First, we have to select $\sigma$.
Second, we should define appropriate $f_i$'s.
Finally, for a given tree $s$, $\mathope{split}_{\sigma}$ takes only $o_{\sigma}(s)$ and $\partial_{\sigma}(s)$
and forgets other information.
Thus, no matter how we select $\sigma$ and $f_i$'s, we cannot fully recover original tree $s$.
So, we focus on recovering the image $f(s)$.

Our solution is taking an index $i \in I$ such that $\iota_i(X) \cap f(S_{\mathbb{T}}) \neq \emptyset$
and judging whether $f(s) \in \iota_i(X)$ or not.
Now $\iota_i(X)$ is in the sub-basis of the copower $I \cdot X$.
Therefore, $f^{-1}(\iota_i(X))$ has the form $[t^{(i)} \mapsto v^{(i)}]$ for some term $t^{(i)}$ and some valuable $v^{(i)}$.
This says that we can judge whether $s \in f^{-1}(\iota_i(X))$ or not by observing {\it the path of $t^{(i)}$ along $s$}.

\begin{definition}
    For a free theory $\mathbb{T}$ and a term $t\in T(V)$,
    the set of {\it paths} of the term $t$, $\mathope{Path}(t)$, is defined as follows.
    Each path is a sequence of pairs $(\sigma,i)$ of a symbol $\sigma$ in $\Sigma_{\mathbb{T}}$ and $i \in |\sigma|$.
    If $t$ is a variable $v$, then the only path of $v$ is the empty sequence.
    A sequence $(\sigma_1,i_1)\cdots(\sigma_n,i_n)$ is in $\mathope{Path}(t)$ if
    $t$ is of the form $\sigma_1(\lambda i.t_i)$ and
    $(\sigma_2,i_2)\cdots(\sigma_n,i_n) \in \mathope{Path}(t_{i_1})$.
\end{definition}

\begin{definition}
    Let $t \in T(V)$ and $s$ be a state in $S_{\mathbb{T}}$.
    Define inductively {\it the path of $t$ along $s$}, $\mathope{path}(t,s) \in \mathope{Path}(t)$:
    \[
        \begin{array}{lll}
            \mathope{path}(v,s) & \coloneqq & \epsilon, \\
            \mathope{path}(\sigma(\lambda i.t_i),s) & \coloneqq &
            (\sigma,o_{\sigma}(s)) \cdot \mathope{path}(t_{o_{\sigma}(s)},\partial_{\sigma}(s)). \\
        \end{array}
    \]
\end{definition}

Obviously, if $\mathope{path}(t^{(i)},s)=\mathope{path}(t^{(i)},s')$ then $f(s) \in \iota_i(X)$ iff $f(s') \in \iota_i(X)$.
Thus, we select $\sigma$ and $f_i$'s as remembering the information about $\mathope{path}(t^{(i)},s)$ for given $s$.

\begin{definition}
    Let $f \in \cat(S_{\mathbb{T}},X)$, $\sigma$ be a symbol and $i \in |\sigma|$.
    Define $f^{\sigma,i} \in \cat(S_{\mathbb{T}},X)$ as, for a given $s \in S_{\mathbb{T}}$,
    constructing a new tree $s'$ such that
    \[o_{\sigma}(s')=i,\ \partial_{\sigma}(s')=s\]
    and applying $f$ to $s'$.
    Concretely, $f^{\sigma,i}(s) \coloneqq f(^{\sigma,i}s)$,
    where $^{\sigma,i}s$ is the state such that
    \[o_{\sigma}(^{\sigma,i}s)=i,\  \partial_{\sigma}(^{\sigma,i}s)=s\]
    and for other symbols $\tau$,
    \[o_{\tau}(^{\sigma,i}s)=o_{\tau}(s),\ \partial_{\tau}(^{\sigma,i}s)=\partial_{\tau}(s).\]

    For a sequence $p=(\sigma_1,i_1)\cdots(\sigma_n,i_n)$,
    $f^p$ denotes the map $(f^{\sigma_1,i_1})^{\cdots \sigma_n,i_n}$ and
    $^ps$ denotes the state $^{\sigma_1,i_1}(^{\cdots \sigma_n,i_n}s)$
    (if $p=\epsilon$ then $f^{\epsilon} \coloneqq f$ and $^{\epsilon}s \coloneqq s$).
    Then $f^p(s)=f(^ps)$.
\end{definition}

\begin{remark}
    For $f \in \cat(S_{\mathbb{T}},X)$, the function $f^{\sigma,i}$ is indeed continuous on sub-basis
    because for a sub-basis $U$ of $X$,
    $s \in (f^{\sigma,i})^{-1}(U) \Leftrightarrow s^{\sigma,i} \in f^{-1}(U)$
    and $s^{\sigma,i}$ behaves much like $s$.
    The nontrivial case is when $f^{-1}(U)=[\sigma(\lambda i.t_i) \mapsto v]$,
    but we can easily see that $\sem{\sigma(\lambda i.t_i)}(s^{\sigma,i})=\sem{t_i}(s)$.
    Thus, in this case, $s \in (f^{\sigma,i})^{-1}(U) \Leftrightarrow s \in [t_i \mapsto v]$.
\end{remark}

\begin{definition}
    For $f \in \cat(S_{\mathbb{T}},X)$ and for a term $t$,
    we define $[f,t] \in \cat(S_{\mathbb{T}},X)$ inductively:
    \begin{align}
        [f,v] & \coloneqq f \\
        [f,\sigma(\lambda i.t_i)] & \coloneqq \mathope{split}_{\sigma}(\lambda i.[f^{\sigma,i},t_i])
    \end{align}
\end{definition}

\begin{example}
    \label{exa:2}
    When $t=\sigma_1(\lambda k.\sigma_2(\lambda l.(k,l)))$,
    for a given state $s$ such that
    $o_{\sigma_1}(s)=k_0$ and $o_{\sigma_2}(\partial_{\sigma_1}(s))=l_0$,
    \begin{align}
        [f,t](s)
        &=[f^{(\sigma_1,k_0)},\sigma_2(\lambda l.(k_0,l))](\partial_{\sigma_1}(s))\\
        &=[f^{(\sigma_1,k_0)(\sigma_2,l_0)},(k_0,l_0)](\partial_{\sigma_2}(\partial_{\sigma_1}(s)))\\
        &=f^{(\sigma_1,k_0)(\sigma_2,l_0)}(\partial_{\sigma_2}(\partial_{\sigma_1}(s)))\\
        &=f^{(\sigma_1,k_0)}(^{(\sigma_2,l_0)}\partial_{\sigma_2}(\partial_{\sigma_1}(s)))\\
        &=f(^{(\sigma_1,k_0)(\sigma_2,l_0)}\partial_{\sigma_2}(\partial_{\sigma_1}(s))).
    \end{align}
    The state $s' \coloneqq {^{(\sigma_1,k_0)(\sigma_2,l_0)}\partial_{\sigma_2}(\partial_{\sigma_1}(s))}$ has
    the same behavior as $s$ on $t$ i.e.
    $o_{\sigma_1}(s')=o_{\sigma_1}(s)=k_0$ and $o_{\sigma_2}(\partial_{\sigma_1}(s'))=o_{\sigma_2}(\partial_{\sigma_1}(s))=l_0$.
    Additionally, $s'$ behaves completely in the same way as $s$ after $t$, i.e.,
    $\partial_{\sigma_2}(\partial_{\sigma_1}(s'))=\partial_{\sigma_2}(\partial_{\sigma_1}(s))$.
\end{example}

When we write the state constructed by $[f,t^{(i)}]$ from a state $s$ as $\lrangle{t^{(i)},s}$,
we can show that $[f,t^{(i)}](s) = f(\lrangle{t^{(i)},s})$ and $f(\lrangle{t^{(i)},s}) = f(s)$.
The former is shown by easy induction on $t^{(i)}$ (we have to describe $\lrangle{t^{(i)},s}$ concretely).
The latter is established as follows (the formal proof is too long to describe here).
\begin{enumerate}[label=$\cdot$]
    \item Write $f(s)$,$f(\lrangle{t^{(i)},s}) \in I \times X$ as $(j_0,x_0)$,$(j_1,x_1)$.
    \item Since $f$ is continuous on sub-basis, the index $j_0$ and behaviors of $x_0$ as a state of $\mathbb{T'}$-comodel $\bm{X}$ are
    determined by behaviors of $s$. Similarly, $j_1$ and behaviors of $x_1$ are determined by $\lrangle{t^{(i)},s}$.
    (That is, if $f^{-1}(\iota_{i_0}(X))=[t \mapsto v]$ then $j_0=i_0 \Leftrightarrow s \in [t \mapsto v]$.
    Behaviors of $x_0$ is examined by observing whether $x \in [t' \mapsto v']$ for some $t',v'$
    and its inverse image $f^{-1}([t' \mapsto v'])$ is also of the form $[t'' \mapsto v'']$.)
    \item With effort, we can show that required behaviors of $s$ and $\lrangle{t^{(i)},s}$ are
    those on $t^{(i)}$ (i.e., we should ask values of $\sem{t^{(i)}}(s)$ and $\sem{t^{(i)}}(\lrangle{t^{(i)},s})$)
    or those after $t^{(i)}$ (i.e., we should observe $\sem{t}(s)$ and $\sem{t}(\lrangle{t^{(i)},s})$ for some $t$ compatible with $t^{(i)}$).
    \item Behaviors of $\lrangle{t^{(i)},s}$ are completely the same as those of $s$ on $t^{(i)}$ and after $t^{(i)}$,
    as described in Example \ref{exa:2}.
    \item Therefore, $j_0=j_1$ and behaviors of $x_0$ and $x_1$ are completely the same.
    \item Since $\bm{X}$ is a simple comodel, the behavior determines the state
    (if $\bm{X}$ is not simple, there may be different states with the same behavior).
\end{enumerate}
Now we get the following lemma:

\begin{lemma}
    \label{lem:1}
    Let $f \in \mathword{Sub}(S_{\mathbb{T}},I \cdot X)$ such that
    $I_f \coloneqq \{i \in I \mid \iota_i(X) \cap f(S_{\mathbb{T}}) \neq \emptyset\}$ has at least two elements.
    Take an index $i \in I_f$,
    an appropriate $\mathbb{T}$-term $t^{(i)}$ and a variable $v^{(i)}$ as $f^{-1}(\iota_i(X)) = [t^{(i)} \mapsto v^{(i)}]$.
    Then $f = [f,t^{(i)}]$.
\end{lemma}
Moreover, by definition, $[f,t^{(i)}]$ consists of $\mathope{split}_{\sigma}$'s and $f^p$'s ($p \in Path(t^{(i)})$).
\begin{lemma}
    \label{lem:2}
    In the situation of Lemma \ref{lem:1}, maps $f^p$'s appearing in $[f,t^{(i)}]$ satisfy $I_{f^p} \subsetneq I_f$
    (where $I_f \coloneqq \{i \in I \mid \iota_i(X) \cap f(S_{\mathbb{T}}) \neq \emptyset\}$).
\end{lemma}
Intuition is that when calculation reaches $f^p$,
it has already identified whether a given state belongs to $f^{-1}(\iota_i(X))$ or not.
Thus $I_{f^p}=\{i\}$ or $i \notin I_{f^p}$.

\vspace{8pt}

So far, we completed the part (v) of the overview.
If $I_{f}$ is finite, by induction,
we can express $f$ by using operations $\mathope{split}_{\sigma}$ and
maps $g$ such that $g:S_{\mathbb{T}} \rightarrow \iota_{i_g}(X)$ for an index $i_g \in I$,
and conclude $f \in M$ (complete (vi) of the overview).

It remains to consider the case that $I_{f}$ is infinite.
First, specify the induction in the finite case (we assume $I \neq \emptyset$ for simplicity):
\begin{enumerate}[label=(\arabic*)]
    \item If $|I_{f}|>1$ then we choose an index $l_0 \in I_f$.
    \item Take a term $t_0$ and a variable $v_0$ as $f^{-1}(\iota_{l_0}(X))=[t_0 \mapsto v_0]$.
    \item Write $f$ as $[f,t_0]$.
    \item Consider $f^p$'s, $p \in \mathope{Path}(t_0)$ (maps appearing in $[f,t_0]$).
    \item By Lemma \ref{lem:2}, for all $p \in \mathope{Path}(t_0)$, we have $I_{f^p} \subsetneq I_{f}$ and they are not empty.
    \item If $|I_{f^p}|=1$ for all $p$ then the induction is finished.
    \item If there are paths $p$ with $|I_{f^p}|>1$,
    we apply this procedure for such maps $f^p$ instead of $f$
    (take $l_1^p \in I_{f^p}$, let $(f^p)^{-1}(\iota_{l_1^p}(X))=[t_1^p \mapsto v_1^p]$, etc.).
\end{enumerate}

We call this procedure the {\it splitting procedure} of $f$.
Even when $I_f$ is infinite,
if we reach the situation such that all maps $f^p$ appearing in the splitting procedure
satisfy the condition that $I_{f^p}$ is finite, then we can conclude $f \in M$.
Thus our aim is to show that we can always reach this situation, by contradiction.

Suppose that we cannot reach this situation.
Then the splitting procedure continues infinitely and
we get a infinite sequence $(\sigma_1,i_n)(\sigma_2,i_2) \cdots$ such that
for all $n \geq 1$, the set $I_{f^{(\sigma_1,i_1)\cdots(\sigma_n,i_n)}}$ is infinite.
Since this sequence is given by the splitting procedure of $f$,
there exists the unique natural number $n_1$ such that the sequence
$p_1 \coloneqq (\sigma_1,i_1) \cdots (\sigma_{n_1},i_{n_1}) \in \mathope{Path}(t_0)$
and we apply (7) to the map $f^{p_1}$; take appropriate $l_1 \in I$ (especially $l_1 \neq l_0$),
$t_1$, $v_1$ as $f_0^{-1}(\iota_{l_1}(X))=[t_1 \mapsto v_1]$
then we get the sequence $p_2 \coloneqq (\sigma_{n_1+1},i_{n_1+1}) \cdots (\sigma_{n_2},i_{n_2}) \in \mathope{Path}(t_1)$
for the unique $n_2$.
Therefore we are in the following situation:
\begin{enumerate}[label=]
    \item $\exists$infinite sequence of natural numbers $0=n_0 < n_1 < n_2 < \cdots$,
    \item $\forall j \geq 0$, $\exists l_j \in I$ (different from each other), $\exists$term $t_j$, $\exists$variable $v_j$,
    \item $p_{j+1} \coloneqq (\sigma_{n_j+1},i_{n_j+1})\cdots(\sigma_{n_{j+1}},i_{n_{j+1}}) \in \mathope{Path}(t_j)$,\label{property:2}
    \item $(f^{q_j})^{-1}(\iota_{l_j}(X))=[t_j \mapsto v_j]\neq \emptyset,S_{\mathbb{T}}$, \label{property:3}
    \item $f^{q_j}=[f^{q_j},t_j]$, \label{property:4}
    \item $I_{f^{q_j}}$ is infinite. \label{property:1}
    \item $(q_j \coloneqq (\sigma_1,i_1)\cdots(\sigma_{n_j},i_{n_j})=p_1 \cdots p_j,$
    $q_0=p_0 \coloneqq \epsilon.)$
\end{enumerate}

If a state $s$ satisfies $\mathope{path}(t_j,s)=p_{j+1}$ then $\sem{t_j}(s) \neq v_j$
i.e. $f^{q_j}(s) \notin \iota_{l_j}(X)$:
because if $\sem{t_j}(s)=v_j$ for this $s$, then for all state $s'$,
$\sem{t_j}(^{p_{j+1}}s')=v_j$ and this implies
$f^{q_{j+1}}(s')=f^{q_j \cdot p_{j+1}}(s)=f^{q_j}(^{p_{j+1}}s') \in \iota_{l_j}(X)$,
contradicts to the assumption that $I_{f^{q_{j+1}}}$ is infinite.

By extending this argument, for a state $s' \coloneqq {^{q_{j+1}}s}$ ($s$ is an arbitrary state),
we can show that $f(s')$ does not belong to any of $\iota_{l_0}(X), \ldots \iota_{l_j}(X)$.
For example, when $j=1$, then $s'={^{p_1 \cdot p_2}s}$ and
it is clear that $f(s') \notin \iota_{l_0}(X)$ by $\mathope{path}(t_0,s')=p_1$;
for $\iota_{l_1}(X)$, we can calculate that $f(s')=[f,t_0](s')=f^{p_1}(^{p_2}s)=f^{q_1}(^{p_2}s)$,
thus $f(s') \in \iota_{l_1}(X)$ iff $f^{q_1}(^{p_2}s) \in \iota_{l_1}(X)$ but the latter is false.

Then if we can construct a state $s'$ such as ${{^{p_1 \cdot p_2 \cdots}}s}$ (this notation is informal),
we expect that $f(s')$ does not belong to $\iota_{l_j}(X)$ for all $j$.
Assume we have such a state $s'$, and let $f(s') \in \iota_l(X)$ for $l \in I$,
which satisfies $l \neq l_j$ for all $j$.
By the construction of the sequence $(\sigma_1,i_1)(\sigma_2,i_2)\cdots$,
we may have to observe an infinite behavior of a given state $s''$ to decide whether $f(s'') \in \iota_l(X)$.
This contradicts to the continuity of $f$ on sub-basis.

For lack of space, we only explain the construction of such a state and the outline of derivation of a contradiction.
\begin{definition}
    For any state $s \in S_{\mathbb{T}}$ and
    any infinite sequence $\overrightarrow{(\tau,i)}=(\tau_1,i_1)(\tau_2,i_2)\cdots$,
    where $\tau_k$ is a operation symbol in $\mathbb{T}$ and $i_k \in |\tau_k|$ for all $k$,
    we define a new state $^{\overrightarrow{(\tau,i)}}s$ as the function
    $\Sigma_{\mathbb{T}}^{\ast} \rightarrow \prod_{\sigma \in \Sigma_{\mathbb{T}}}|\sigma|$:
    \begin{equation}
        ^{\overrightarrow{(\tau,i)}}s(\overrightarrow{x}) \coloneqq
        \begin{cases}
            s(\epsilon)[\tau_k \rightarrow i_k] & (\overrightarrow{x}=\tau_1 \cdots \tau_{k-1})\\
            s(\overrightarrow{y}) & (\overrightarrow{x} = \tau_1\cdots\tau_{k-1}\overrightarrow{y},\ y_1 \neq \tau_k)\ (k\geq 1)
        \end{cases}
    \end{equation}
\end{definition}

We define the infinite sequence $\overrightarrow{(\tau,i)_k}$ for a natural number $k \geq 1$ and
a infinite sequence $\overrightarrow{(\tau,i)}=(\tau_n,i_n)_{n \geq 1}$ as:
\[
    \overrightarrow{(\tau,i)_k}  \coloneqq  (\tau_n,i_n)_{n \geq k}.
\]

\begin{lemma}
    For any state $s \in S_{\mathbb{T}}$ and
    any infinite sequence $\overrightarrow{(\tau,i)}$,
    \begin{equation}
        \label{eq:1}
        o_{\tau_1}(^{\overrightarrow{(\tau,i)}}s)=i_1,\ 
        \partial_{\tau_1}(^{\overrightarrow{(\tau,i)}}s)={^{\overrightarrow{(\tau,i)_2}}}s,
    \end{equation}
    and inductively, for each $k \geq 1$,
    \begin{align}
            o_{\tau_k}(\partial_{\tau_{k-1}}\cdots\partial_{\tau_1}(^{\overrightarrow{(\tau,i)}}s))&=i_k \label{eq:2} \\
            \partial_{\tau_{k}}\cdots\partial_{\tau_1}(^{\overrightarrow{(\tau,i)}}s)&=^{\overrightarrow{(\tau,i)_{k+1}}}s. \label{eq:3}
    \end{align}
\end{lemma}

Consider the infinite sequence $\overrightarrow{(\sigma,i)}$ constructed by
the splitting procedure of $f_0$ and fix a state $s$.
When we abbreviate the composition $\partial_{\sigma_{n_j}} \cdots \partial_{\sigma_{n_1}}$
as $\partial_j$ ($\partial_0$ is the identity on $S_{\mathbb{T}}$),
then by the above lemma, for all $j$,
\begin{equation}
    \label{eq:8}
    \mathope{path}(t_j,\partial_j(^{\overrightarrow{(\sigma,i)}}s))=
    (\sigma_{n_j+1},i_{n_j+1}) \cdots (\sigma_{n_{j+1}},i_{n_{j+1}})
    =p_{j+1}.
\end{equation}
This implies $f(^{\overrightarrow{(\sigma,i)}}s) \notin \iota_{l_j}(X)$ for all $j$:
we can show by easy induction on $j$ that
$f(^{\overrightarrow{(\sigma,i)}}s)=f^{q_j}(\partial_j(^{\overrightarrow{(\sigma,i)}}s))$
for each $j$,
and then by \eqref{eq:8},
$f^{q_j}(\partial_j(^{\overrightarrow{(\sigma,i)}}s)) \notin \iota_{l_j}(X)$.

Let $l \in I$ be the index such that $f(^{\overrightarrow{(\sigma,i)}}s) \in \iota_l(X)$,
then $l \neq l_j$ for all $j$.
By the continuity of $f$, $f^{-1}(\iota_l(X))=[t \mapsto v]$ for some $t$ and $v$.
We show, by induction on $j \geq 1$, that each $q_j$ is a prefix of
the path $\mathope{path}(t,{^{\overrightarrow{(\sigma,i)}}s})$ and
therefore the length of $\mathope{path}(t,{^{\overrightarrow{(\sigma,i)}}s})$ is infinite,
this is irrational.

We have reached the contradiction.
Thus, even if $I_{f}$ is infinite,
the splitting procedure of $f$ comes down to the finite case
and therefore $f \in M$.
This implies $M=\cat(S_{\mathbb{T}},X)$,
we complete the proof of Theorem \ref{main thm}.

\vspace{10pt}

To complete the proof of that the functor $\cat(\bm{S}_{\mathbb{T}},\_)$ preserves copowers of
simple comodels (the part (II) of the overview),
we show the existence of the mediating morphism.
In the proof below, we use the fact that $M=\cat(S_{\mathbb{T}},X)$ we have shown above
and thus we assume $X \in \cat$ has a simple $\mathbb{T'}$-comodel structure.

\begin{proposition}[the existence of mediating morphism]
    Assume that $X \in \cat$ has a simple $\mathbb{T'}$-comodel structure.
    If there is a family of $\mathbb{T}$-model map
    \[(p_i:\cat(\bm{S}_{\mathbb{T}},X) \rightarrow \bm{Y})_{i \in I}\]
    then there is a $\mathbb{T}$-model map
    \[\overline{p}:\cat(\bm{S}_{\mathbb{T}},I \cdot X) \rightarrow \bm{Y}\]
    such that $p_i=\overline{p} \circ (\iota_i \circ (\_))$ for all $i \in I$.
\end{proposition}

\begin{proof}
    Let $\bm{N}$ be the free $\mathbb{T}$-model on the set $\cat(S_{\mathbb{T}},X) \times I$.
    By freeness, we have the unique $\mathbb{T}$-model map $\beta:\bm{N} \rightarrow \cat(\bm{S}_{\mathbb{T}},X)$
    with $\beta(f,i)=\iota_i f$
    and the unique map $p':\bm{N} \rightarrow \bm{Y}$ with $p'(f,i)=p_i(f)$.
    It suffices to show that there is a factorization $p$ of $p'$ through $\beta$ i.e. $p'=p \circ \beta$.
    By the proof of the above theorem, $\beta$ is epimorphic.
    So, it suffices to show that if $x,y \in N$ satisfy $\beta(x)=\beta(y)$,
    then they satisfy $p'(x)=p'(y)$.
    We can do this by induction on the total number of operation symbols in $\mathbb{T}$-terms $x$, $y$.
\end{proof}

Now we have completed all of the steps described in the overview and reached the goal:

\begin{theorem}
    For free theories $\mathbb{T}$, $\mathbb{T'}$,
    the set $\cat(S_{\mathbb{T}},S_{\mathbb{T'}})$
    appears as the final $\mathbb{T}\-\mathbb{T'}$-bimodel
    with $\mathbb{T}$-model structure maps $(\mathope{split}_{\sigma})_{\sigma \in \Sigma_{\mathbb{T}}}$ and
    $\mathbb{T'}$-comodel structure maps
    \begin{equation}
        \cat(\bm{S}_{\mathbb{T}},S_{\mathbb{T'}})
        \xrightarrow{\sem{\tau}^{\bm{S}_{\mathbb{T'}}} \circ (\_)}
        \cat(\bm{S}_{\mathbb{T}},|\tau| \cdot S_{\mathbb{T'}})
        \xrightarrow{\cong}
        |\tau| \cdot \cat(\bm{S}_{\mathbb{T}},S_{\mathbb{T'}})
    \end{equation}
    for $\tau \in \Sigma_{\mathbb{T'}}$,
    where the first part is postcomposition with
    the $\mathbb{T'}$-comodel structure map $\sem{\tau}^{\bm{S}_{\mathbb{T'}}}$ of
    the final $\mathbb{T'}$-comodel $\bm{S}_{\mathbb{T'}}$
    and the second part is the canonical isomorphism coming from the fact that
    $\cat(\bm{S}_{\mathbb{T}},\_):\cat \rightarrow Mod(\mathbb{T})$
    preserves copowers.
\end{theorem}

Finally, we describe the relation between the final residual comodel $I_{\mathbb{T},\mathbb{T'}}$ and the final bimodel $\cat(S_{\mathbb{T}},S_{\mathbb{T'}})$.
In section \ref{sec:4}, we constructed
$\textsf{reflect}:I_{\mathbb{T},\mathbb{T'}} \rightarrow \cat(S_{\mathbb{T}},S_{\mathbb{T'}})$.
We can also define a map of converse direction
$\textsf{reify}:\cat(S_{\mathbb{T}},S_{\mathbb{T'}}) \rightarrow I_{\mathbb{T},\mathbb{T'}}$.
This is because each $\mathbb{T}$-$\mathbb{T'}$-bimodel equips a $\mathbb{T}$-residual $\mathbb{T'}$-comodel structure.
We can show that the composition $\textsf{reflect} \circ \textsf{reify}$ is
the identity on $\cat(S_{\mathbb{T}},S_{\mathbb{T'}})$.
This establishes the complete representation of $\cat(S_{\mathbb{T}},S_{\mathbb{T'}})$ by $I_{\mathbb{T},\mathbb{T'}}$.
The argument is the same as \cite{garnerStreamProcessorsComodels2021}.

\section{Conclusion and Future Work}
\label{sec:8}

Our main contribution is giving a comodel-theoretic characterization of
$\cat(S_{\mathbb{T}},S_{\mathbb{T'}})$, which is a subset of
$\mathword{Top}(S_{\mathbb{T}},S_{\mathbb{T'}})$.
Explicitly, $\cat(S_{\mathbb{T}},S_{\mathbb{T'}})$ can
appear as the final $\mathbb{T}$-$\mathbb{T'}$-bimodel.
Additionally, we constructed maps
\[
    \begin{tikzpicture}[auto]
        \node (C) at (0,0) {$\cat(S_{\mathbb{T}},S_{\mathbb{T'}})$}; \node (D) at (3,0) {$I_{\mathbb{T},\mathbb{T'}}$,};
        \draw[<-, transform canvas={yshift=6pt}] (C) --node {$\textsf{reflect}$}(D);
        \draw[->, transform canvas={yshift=-6pt}] (C) --node[swap] {$\textsf{reify}$} (D);
    \end{tikzpicture}
\]
where $I_{\mathbb{T},\mathbb{T'}}$ is the underlying set of the final $\mathbb{T}$-residual $\mathbb{T'}$-comodel
and we identify $\cat(S_{\mathbb{T}},S_{\mathbb{T'}})$ with the final $\mathbb{T}$-$\mathbb{T'}$-bimodel.
Then we showed that the composition ${\textsf{reflect} \circ {\textsf{reify}}}$ is
the identity on $\cat(S_{\mathbb{T}},S_{\mathbb{T'}})$,
this implies the completeness of $\textsf{reflect}$.
A further generalization is required in order to give a complete representation of
$\mathword{Top}(S_{\mathbb{T}},S_{\mathbb{T'}})$.
In \cite{article}, Ghani et al. gave a coalgebraic representation of continuous functions
on final coalgebras of various functors but they did not show its completeness.
On the other hand, although we show a kind of complete correspondence,
this does not consider the whole set of continuous functions.

Our second contribution is an analysis of the final comodel of a free algebraic theory and
functions continuous on sub-basis from it.
During the proof in Section \ref{sec:6},
we defined several notions such as
the continuous function $f^{\sigma,i}$ and the state $^{\sigma,i}s$.
We expect that they and the ideas underlying them are useful when studying arbitrary continuous functions or
investigating the case of non-free algebraic theories.

According to \cite{lmcs:7712},
transducers with backtracking characterize continuous functions between the set of trees.
After submitting this article, we established a retraction between
appropriate transducers and the residual comodels.
We hope to report this result elsewhere.

When we try to generalize our argument to the case of non-free algebraic theories, there are many difficulties.
One is the question whether bimodels can be seen as residual comodels (the latter of Lemma \ref{lem:bimodel}).
This is a key point to define the map $\textsf{reify}$.
Of course, there can be unnoticed issues.
We should carefully analyse our proofs and
we would like to identify algebraic theories in which our argument is effective.

\section*{Acknowledgements}

I am grateful to Soichiro Fujii for suggesting the topic treated in this paper.
He and my supervisor Masahito Hasegawa helped me in many ways,
and discussion with them were very meaningful.
Finally, I would like to thank my advisors and colleagues for broadening and deepening my knowledge.

\bibliography{shuron1}

\end{document}